\documentclass[11pt]{article}%

\ifx\SODAVer\undefined%
\newcommand{\SODA}[1]{}
\newcommand{\notSODA}[1]{#1}
\else
\newcommand{\SODA}[1]{#1}
\newcommand{\notSODA}[1]{}
\fi

\SODA{%
   \def\CiteCrazy{1}%
}%

\def\UseBibLatex{1}

\makeatletter
\def\input@path{{styles/}}
\makeatother

\providecommand{\BibLatexMode}[1]{}
\providecommand{\BibTexMode}[1]{}

\renewcommand{\BibLatexMode}[1]{#1}
\renewcommand{\BibTexMode}[1]{}

\ifx\UseBibLatex\undefined%
  \renewcommand{\BibLatexMode}[1]{}
  \renewcommand{\BibTexMode}[1]{#1}
\fi

\BibLatexMode{%
   \usepackage[bibencoding=utf8,style=alphabetic,backend=biber]{biblatex}%
   \usepackage{sariel_biblatex}%
   \usepackage{biblatex_authors}
}

\usepackage{amsmath}%
\usepackage{amssymb}%
\usepackage{mathtools}
\usepackage[table]{xcolor}%
\usepackage{wrapfig}

\usepackage{salgorithm}%

\usepackage[amsmath,thmmarks]{ntheorem}%

\usepackage{titlesec}%
\usepackage{xcolor}%
\usepackage{mleftright}%
\usepackage{xspace}%
\usepackage{graphicx}
\usepackage{hyperref}%
\usepackage[inline]{enumitem}
\usepackage{hyperref}%
\usepackage[ocgcolorlinks]{ocgx2}

\usepackage{mathcalb}

\usepackage{microtype}

\usepackage{caption}

\hypersetup{%
      unicode,
      breaklinks,%
      colorlinks=true,%
      urlcolor=[rgb]{0.25,0.0,0.0},%
      linkcolor=[rgb]{0.5,0.0,0.0},%
      citecolor=[rgb]{0,0.2,0.445},%
      filecolor=[rgb]{0,0,0.4},
      anchorcolor=[rgb]={0.0,0.1,0.2}%
}

\titlelabel{\thetitle. }%

\theoremseparator{.}%

\theoremstyle{plain}%
\newtheorem{theorem}{Theorem}[section]

\newtheorem{lemma}[theorem]{Lemma}

\newtheorem{corollary}[theorem]{Corollary}

\newtheorem{observation}[theorem]{Observation}

\theoremstyle{plain}%
\theoremheaderfont{\sf} \theorembodyfont{\upshape}%
\newtheorem*{remark:unnumbered}[theorem]{Remark}%
\newtheorem{remark}[theorem]{Remark}%

\newtheorem{defn}[theorem]{Definition}

\theoremheaderfont{\em}%
\theorembodyfont{\upshape}%
\theoremstyle{nonumberplain}%
\theoremseparator{}%
\theoremsymbol{\myqedsymbol}%
\newtheorem{proof}{Proof:}%

\definecolor{blue25emph}{rgb}{0, 0, 0.4}
\definecolor{red25emph}{rgb}{0.4, 0, 0}

\newcommand{\emphi}[1]{\emphw{#1}}

\definecolor{almostblack}{rgb}{0, 0.0, 0.3}

\providecommand{\emphw}[1]{}%
\renewcommand{\emphw}[1]{{\textcolor{almostblack}{\emph{#1}}}}%

\providecommand{\emphOnly}[1]{}%
\renewcommand{\emphOnly}[1]{\emph{\textcolor{blue25emph}{\textbf{#1}}}}

\newcommand{\myqedsymbol}{\rule{2mm}{2mm}}
\newcommand{\SarielThanks}[1]{%
   \thanks{%
      School of Computing and Data Science; %
      University of Illinois; %
      201 N. Goodwin Avenue; %
      Urbana, IL, 61801, USA; %
      \href{mailto:spam@illinois.edu}{sariel@illinois.edu}; %
      \url{http://sarielhp.org/}.%
   #1%
   }%
}

\newcommand{\HLink}[2]{\hyperref[#2]{#1~\ref*{#2}}}
\newcommand{\HLinkY}[2]{\hyperref[#2]{#1}}
\newcommand{\HLinkSuffix}[3]{\hyperref[#2]{#1\ref*{#2}{#3}}}

\newcommand{\figlab}[1]{\label{fig:#1}}
\newcommand{\figref}[1]{\HLink{Figure}{fig:#1}}

\newcommand{\thmlab}[1]{{\label{theo:#1}}}
\newcommand{\thmref}[1]{\HLink{Theorem}{theo:#1}}

\providecommand{\deflab}[1]{\label{def:#1}}
\newcommand{\defref}[1]{\HLink{Definition}{def:#1}}
\newcommand{\defrefY}[2]{\HLinkY{#2}{def:#1}}

\newcommand{\corlab}[1]{\label{cor:#1}}
\newcommand{\corref}[1]{\HLink{Corollary}{cor:#1}}%

\newcommand{\obslab}[1]{\label{observation:#1}}
\newcommand{\obsref}[1]{\HLink{Observation}{observation:#1}}

\newcommand{\tbllab}[1]{\label{table:#1}}
\newcommand{\tblref}[1]{\HLink{Table}{table:#1}}

\newcommand{\lemlab}[1]{\label{lemma:#1}}
\newcommand{\lemref}[1]{\HLink{Lemma}{lemma:#1}}%

\newcommand{\itemlab}[1]{\label{item:#1}}
\newcommand{\itemref}[1]{\HLinkSuffix{}{item:#1}{}}

\providecommand{\eqlab}[1]{}%
\renewcommand{\eqlab}[1]{\label{equation:#1}}

\newcommand{\seclab}[1]{\label{sec:#1}}
\newcommand{\secref}[1]{\HLink{Section}{sec:#1}}

\providecommand{\eqlab}[1]{}%
\renewcommand{\eqlab}[1]{\label{equation:#1}}
\newcommand{\Eqref}[1]{\HLinkSuffix{Eq.~(}{equation:#1}{)}}

\providecommand{\remove}[1]{}%
\newcommand{\Set}[2]{\left\{ #1 \;\middle\vert\; #2 \right\}}

\newcommand{\pth}[1]{\mleft(#1\mright)}%

\newcommand{\cardin}[1]{\left\lvert {#1} \right\rvert}%

\renewcommand{\th}{-th\xspace}
\renewcommand{\th}{th\xspace}
\renewcommand{\th}{\ensuremath{\hphantom{}^{\,\mathrm{th}}}\xspace}

\newcommand{\ZZ}{\mathbb{Z}}%
\renewcommand{\Re}{\mathbb{R}}%
\newcommand{\N}{\mathbb{N}}

\newlist{compactenumA}{enumerate}{5}%
\setlist[compactenumA]{itemsep=-0.5ex,topsep=0.5ex,partopsep=1ex,parsep=1ex,%
   label=(\Alph*)}%

\newlist{compactenuma}{enumerate}{5}%
\setlist[compactenuma]{itemsep=-0.5ex,topsep=0.5ex,partopsep=1ex,parsep=1ex,%
   label=(\alph*)}%

\newlist{compactenumI}{enumerate}{5}%
\setlist[compactenumI]{itemsep=-0.25ex,topsep=0.7ex,partopsep=1ex,parsep=1ex,%
   label=(\Roman*)}%

\newlist{compactenumi}{enumerate}{5}%
\setlist[compactenumi]{itemsep=-0.5ex,topsep=0.5ex,partopsep=1ex,parsep=1ex,%
   label=(\roman*)}%

\newlist{compactenumi*}{enumerate*}{5}%
\setlist[compactenumi*]{label=(\roman*)}%

\newlist{compactitem}{itemize}{5}%
\setlist[compactitem]{itemsep=-0.5ex,topsep=0.5ex,partopsep=1ex,parsep=1ex,%
   label=\ensuremath{\bullet}}%

\providecommand{\etal}{\textit{et~al.}\xspace}

\numberwithin{figure}{section}%
\numberwithin{table}{section}%
\numberwithin{equation}{section}%

\usepackage{stmaryrd}%

\newcommand{\IRX}[1]{\left[#1 \right]}%
\newcommand{\IRY}[2]{\left[#1 : #2 \right]}%

\newcommand{\IOCX}[1]{\left( #1 \right]}%
\newcommand{\ICOX}[1]{\left[ #1 \right)}%

\newcommand{\suppX}[1]{{\mathsf{S}}\pth{#1}}

\newcommand{\LL}{\mathcal{L}}%

\newcommand{\dY}[2]{\left\| #1 - #2 \right\|}

\newcommand{\normX}[1]{\left\| #1 \right\|}
\newcommand{\simpY}[2]{\mathrm{simp}_{#1}\pth{#2}}

\newcommand{\F}{\mathcal{F}}%

\newcommand{\distC}{\mathcalb{d}}
\newcommand{\dFmC}{\distC_m}
\newcommand{\dFmY}[2]{\dFmC\pth{#1,#2}}

\newcommand{\DistFrechetC}{d_F}

\newcommand{\dFY}[2]{\DistFrechetC\pth{#1,#2}}

\newcommand{\dFC}{\DistFrechetC}
\newcommand{\dFdY}[2]{\dFC\pth{#1,#2}} %

\newcommand{\passX}[1]{\Gamma_{#1}}
\newcommand{\passY}[2]{\Gamma(#1, #2)}

\newcommand{\rad}{\delta}%

\newcommand{\G}{\mathsf{G}}%
\newcommand{\eps}{\varepsilon}%
\newcommand{\epsA}{\xi}%
\newcommand{\Frechet}{Fr\'{e}chet\xspace}

\newcommand{\BFS}{\Term{BFS}\index{BFS}\xspace}

\newcommand{\Term}[1]{\textsf{#1}}

\newcommand{\SETH}{\Term{SETH}\xspace}

\newcommand{\DAG}{\Term{DAG}\xspace}

\newcommand{\gapX}[1]{\mathrm{gap}\pth{#1}}%

\newcommand{\vOpt}{v^\star}

\newcommand{\decider}{\AlgorithmI{decider}\xspace}

\newcommand{\fseq}{sequential\xspace}

\renewcommand{\fseq}{marching\xspace}

\newcommand{\spread}{\Psi}
\newcommand{\spreadX}[1]{\spread\pth{#1}}

\usepackage{cancel}

\usepackage{fullpage}

\newcommand{\lotte}[1]{\textcolor{orange}{Lotte: #1}}
\newcommand{\marena}[1]{\textcolor{violet}{Marena: #1}}
\newcommand{\sariel}[1]{\textcolor{blue!25}{Sariel: #1}}
\newcommand{\anne}[1]{\textcolor{olive}{Anne: #1}}

\usepackage[most]{tcolorbox}
\usepackage{marginnote}

\newcommand{\NewAuthor}[3]{%
    \expandafter\newcommand\csname #1margin\endcsname[1]{%
        \marginnote{\begin{tcolorbox}[
            enhanced, sharp corners, boxrule=0.5pt,
            colback=#2!10, colframe=#2!80!black,
            fontupper=\tiny\sffamily, width=\marginparwidth,
            left=2pt, right=2pt, top=2pt, bottom=2pt
        ] \textbf{#3:} ##1 \end{tcolorbox}}%
    }

    \newenvironment{#1block}[1]
    {\begin{tcolorbox}[
        colback=#2!5, colframe=#2!75!black,
        fonttitle=\bfseries, title=#1's Note: ##1,
        breakable, enhanced
    ]}
    {\end{tcolorbox}}
}

\NewAuthor{sariel}{blue}{SH}
\NewAuthor{lotte}{violet}{LB} %
\NewAuthor{marena}{orange}{MR} %
\NewAuthor{anne}{olive}{AD} %

\renewcommand{\marena}[2][]{\begin{marenablock}{#1}
       #2
   \end{marenablock}
}

\renewcommand{\lotte}[2][]{\begin{lotteblock}{#1}
       #2
   \end{lotteblock}
}

\renewcommand{\sariel}[2][]{\begin{sarielblock}{#1}
       #2
   \end{sarielblock}
}

\renewcommand{\anne}[2][]{\begin{anneblock}{#1}
       #2
   \end{anneblock}
}

\usepackage[ruled, linesnumbered, vlined]{algorithm2e}
\DontPrintSemicolon

\newcommand{\distY}[2]{D({#1}, {#2})}   %

\newcommand{\A}{\textup{A}}     %
\newcommand{\B}{\textup{B}}     %

\newcommand{\elev}{{\gamma}}
\newcommand{\elevS}{\widehat{\elev}}

\newcommand{\elevY}[2]{\elev \pth{#1, #2}}
\newcommand{\esX}[1]{\elevS\pth{#1}}

\newcommand{\Tp}{T_P}           %
\newcommand{\Tq}{T_Q}           %
\newcommand{\Tbsp}{\mathcal{T}}          %
\newcommand{\Tcont}{\Tbsp}          %

\newcommand{\Decomp}{\mathcal{M}}    %

\newcommand{\cDecomp}{\contX{\Decomp}}%

\newcommand{\rect}{\mathsf{r}}%

\newcommand{\rectY}[2]{\rect^{}_{#1,#2}}%

\newcommand{\Neighbors}[1]{\mathcal{N}(#1)}

\newcommand{\cP}{\contX{P}}
\newcommand{\cQ}{\contX{Q}}

\providecommand{\contX}[1]{\overline{#1}}

\newcommand{\repX}[1]{\psi_{#1}}

\newcommand{\closureX}[1]{\mathrm{cl}\pth{#1}}

\newcommand{\WSPD}{\Term{WSPD}\xspace}
\newcommand{\eWSPD}{\ensuremath{\tfrac{1}{\eps}}-\WSPD{}\xspace}

\newcommand{\rootX}[1]{\mathrm{root}\pth{#1}}%

\newcommand{\dmX}[1]{\Delta\pth{#1}}
\newcommand{\dsY}[2]{\bar{d}\pth{#1, #2}}

\newcommand{\cPY}[2]{\cP_{#1:#2}}%
\newcommand{\cQY}[2]{\cQ_{#1:#2}}

\newcommand{\VX}[1]{\sigma_{#1}}
\newcommand{\cVX}[1]{\contX{\sigma}_{#1}}

\newcommand{\segX}[1]{\overline{\mathsf{s}}_{#1}}

\newcommand{\IVX}[1]{I_{#1}}

\newcommand{\XP}{\mathcal{X}}%
\newcommand{\cXP}{\contX{\mathcal{X}}}%
\newcommand{\YQ}{\mathcal{Y}}%
\newcommand{\cYQ}{\contX{\mathcal{Y}}}%

\newcommand{\EXP}{\EitherX{\XP}}
\newcommand{\EYQ}{\EitherX{\YQ}}

\newcommand{\InodeX}[1]{I\pth{#1}}

\newcommand{\first}[1]{{#1}.first}
\newcommand{\last}[1]{{#1}.last}

\usepackage{booktabs} %
\usepackage{tabularx} %

\usepackage{silence}
\newcommand{\vp}{\ensuremath{\mathsf{p}}}
\newcommand{\vq}{\ensuremath{\mathsf{q}}}

\usepackage{stmaryrd}%

\newcommand{\parentX}[1]{#1^{\shortuparrow}}
\newcommand{\pX}[1]{\parentX{#1}}

\providecommand{\contX}[1]{\overline{#1}}%

\newcommand{\ballC}{{B}}%
\newcommand{\ballY}[2]{\ballC\pth{#1,#2}}

\newcommand{\lenX}[1]{\mathsf{L}\pth{ #1 }}%

\newcommand{\remlab}[1]{\label{rem:#1}}
\newcommand{\remref}[1]{\HLink{Remark}{rem:#1}}%

\newcommand{\alignment}{\nu}

\newcommand{\EitherX}[1]{\widetilde{#1}}
\newcommand{\EP}{\EitherX{P}}
\newcommand{\EQ}{\EitherX{Q}}

\newcommand{\cPQ}{\overline{\Pi}}
\newcommand{\PQ}{\Pi}

\newcommand{\OO}{\mathcal{O}}

\newcommand{\Kunnemann}{K\"{u}nnemann\xspace}

\newcommand{\NoAuthors}[1]{}

\newcommand{\Path}{S}%
\newcommand{\aprxPath}{\widehat{S}}

\SODA{%
   \renewcommand{\lotte}[1]{}%
   \renewcommand{\marena}[1]{}%
   \renewcommand{\sariel}[1]{}%
   \renewcommand{\anne}[1]{}%
}

\bibliography{quick_dog}

\begin{document}

\SODA{\setcounter{page}{0}}

\title{The Quick Dog Jumps the Log}

\notSODA{%
   \author{%
      Lotte Blank%
      \thanks{University of Bonn. Funded by the Deutsche Forschungsgemeinschaft (DFG, German Research Foundation) – 459420781 (FOR AlgoForGe).}
      \and%
      Anne Driemel%
      \thanks{University of Bonn and Lamarr Institute for Machine Learning and Artificial Intelligence.}
      \and%
      Sariel Har-Peled%
      \SarielThanks{}
      \and%
      Marena Richter%
      \thanks{University of Bonn. Funded by the Deutsche Forschungsgemeinschaft (DFG, German Research Foundation) – 459420781 (FOR AlgoForGe).}}  }%

\date{\today}

\maketitle

\begin{abstract}
    We give linear-time, and thus optimal, $(1+\eps)$-approximation algorithms for numerous variants of the \Frechet distance between $c$-packed curves (where $c \in O(1)$), removing an additional log factor that was present in previous algorithms. The key to our new algorithms is a linear-size approximation of the elevation function, which uses a decomposition of the domain into rectangles, and a careful implicit dynamic programming on this decomposition. The algorithm extends to the strong, weak, discrete, and continuous \Frechet distances with a running time of roughly $O(cn/\eps)$. The $c$-packedness assumption is used only in the analysis, and the algorithm is simple and should work efficiently for other inputs.
\end{abstract}

\SODA{%
   \thispagestyle{empty}
   \newpage
   \setcounter{page}{1}%
}

\section{Introduction}

In his seminal work, \Frechet \cite{f-sqpdc-06} defined a distance measure between curves in Euclidean $\Re^3$ (in the same paper, he introduced the notion of metric spaces). {\citeau{ag-cfdbt-95}} provided an algorithm for computing this distance between polygonal curves, and it is widely used as a standard metric between curves \cite{tbsbs.ea-catsm-21}. \citeau{ag-cfdbt-95} gave the following analogy for the \Frechet distance: Imagine walking a dog with a fixed-length leash, where the owner is walking along one curve, and the dog is walking along the other curve.  Both have to walk along their respective curve from the beginning to the end while coordinating to stay close together.  The length of the shortest possible leash that permits such a walk is the \Frechet distance between the two curves.

\citeau{ag-cfdbt-95} showed how to compute the distance measure for polygonal curves with a total of $n$ vertices in $\OO(n^2 \log n)$ time by considering a path-finding problem in the sublevel set of the distance function defined on the joint parametric space of the two curves. For this, the critical threshold value when such a path exists is the \Frechet distance, and the algorithm performs a parametric search for this value. Most follow-up work uses a similar framework of performing an implicit or explicit binary search over the set of critical values, thereby retaining a logarithmic factor in the running time of the decision algorithm. Alternatively, one can compute the distance function over the parametric space explicitly and propagate through this space the minimal elevation (the leash length) required to reach that point. Following this approach, one can obtain a $(1+\eps)$-approximation algorithm that runs in $O(n^2 \log{\frac{1}{\eps}})$ time \citeau{bblmm-cfdrl-16}, thus removing the logarithmic factor in $n$.

To avoid the quadratic running time, various approximation algorithms were suggested for special cases of curves, such as approximate shortest paths and locally bounded curves \cite{ahkww-fdcr-06, akw-cdmpc-04}.  The authors of Driemel \etal \cite{dhw-afdrc-10,dhw-afdrc-12} introduced the notion of packedness: a curve is $c$-packed if its length inside any ball is at most $c$ times the radius of the ball. They showed that a $(1+\eps)$-approximation to the \Frechet distance can be computed in $\OO(cn/\eps+cn\log n)$ time for $c$-packed curves. This was later improved to $\OO(\frac{cn}{\sqrt{\eps}}\log \tfrac{1}{\eps} +cn\log n)$ by Bringmann and \Kunnemann \cite{bk-iafdc-17} which is (essentially) tight in the parameters $c$ and $\eps$, assuming \SETH (i.e., strong exponential time hypothesis).

Our work revisits the setting of $c$-packed curves and aims to remove the logarithmic factor in $n$ from the running time. The resulting running time is linear in $n$ and thus optimal in this setting for constant $c$. Our proposed algorithms avoid passing to a decision algorithm by solving the path-finding problem on an approximated elevation function, which is computed based on a partition that is reminiscent of a bipartite well-separated pair decomposition. We show how to apply our new techniques to several prevalent variants of the \Frechet distance, obtaining a similar running time in all cases. Since $c \in O(n)$ for any polygonal curve of $n$ edges, our algorithms are also competitive with the best known algorithms in the general case, while at the same time being conceptually simple.

\subsection{Related work}
\seclab{related_work}

We review the related work for computing the \Frechet distance of two polygonal curves of $n$ vertices in total.  While there has been some recent work on the asymmetric case~\cite{b-fdic-26}, we focus our discussion on the symmetric setting where both curves have $\Omega(n)$ vertices.  For the exact bounds in the asymmetric case, please refer to the original papers.  The currently best-known algorithm for computing the (strong continuous) \Frechet distance of two polygonal curves is by Cheng and Huang~\cite{ch-fdst-25}. Their algorithm runs in $\OO(n^2 \pth{\log\log n}^{2+\mu}/\log^{\mu} n)$ expected time for some constant $\mu \in (0,1)$, which is barely subquadratic.  \Authors{aaks-cdfds-14} showed the first (barely) subquadratic-time algorithm for any variant of the \Frechet distance.  Specifically, they showed that one can compute the (strong) discrete \Frechet distance in $\OO(\frac{n^2\log\log n}{\log n})$ time.  It is believed that there is no strongly subquadratic time algorithm, since the existence of a strongly subquadratic algorithm would refute the Strong Exponential Time Hypothesis (SETH), as was shown by Bringmann~\cite{b-wwdtt-14}. This also holds for any $\alpha$-approximation algorithms for $\alpha < 3$~\cite{bos-sswfd-18}.  Recently, Cheng, Huang, and Zhang \cite{chz-cafds-25} showed that there does exist a randomized strongly subquadratic-time constant-factor $(7+\eps)$-approximation algorithm for the \Frechet distance. The running time is in $\OO(n^{1.99}\log \tfrac{n}{\eps} )$.  For larger approximation factors, algorithms have been proposed with running-time tradeoffs of the form $\OO\pth{\frac{\mathrm{poly}(n)}{\mathrm{poly}(\alpha)}}$ \cite{hkos-sn-cf-23, cf-acfd-21, hkos-ffdat-25, cr-iaadf-18}.

After \citeau{dhw-afdrc-12} introduced the notion of $c$-packedness, there was a sequence of follow-up papers, improving the running time in specific settings. As mentioned above, Bringmann and \Kunnemann~\cite{bk-iafdc-17} improved the running time dependency on $\eps$, which becomes relevant in the range $\frac{1}{\eps} \gg \log n$. If only one of the two polygonal curves is $c$-packed, one can compute a $(1+\eps)$-approximation in $\OO(\frac{cn}{\eps}\log \tfrac{n}{\eps} )$ \cite{gmw-afdwo-24, cvho-cfdwj-26}.

These algorithms still retain the logarithmic factors in $n$ from calls to the decision algorithm.  Our general approach of solving the shortest path problem on the distance terrain, and thereby avoiding such calls to the decision algorithm, has previously been used by the authors of~\citeau{bblmm-cfdrl-16}.  Concretely, their algorithm computes the \Frechet distance under several polyhedral distance functions in $O(n^2)$ time. In the standard setting of the Euclidean distance, their approach does not improve upon the known exact algorithms, but yields a $(1+\eps)$-approximation in $O(n^2 \log \frac{1}{\eps})$. However, the running time of the algorithm is still quadratic for $c$-packed curves, even if $c \ll n$.

A bipartite well-separated pair decomposition for $c$-packed curves, as is used in our paper, has previously been proposed by the authors of \cite{ypfa-seaad-16} and \cite{afpy-adtwe-16} for approximating a different distance measure, namely the Dynamic time warping distance (DTW). However, their algorithms do not avoid logarithmic factors in $n$.  Concretely, the results are $(1+\eps)$-approximation algorithms using $\OO(\frac{nc^2}{\eps}\log\sigma)$ time for two curves with spread $\sigma$ \cite{ypfa-seaad-16} and $\OO(\frac{c n}{\eps}\log n)$ time if the spread is not bounded \cite{afpy-adtwe-16}.  They also study $\kappa$-bounded curves and obtain $(1+\eps)$-approximation algorithms with running time in $\OO(\frac{\kappa^d}{\eps^d}n\log n)$.  We discuss and compare our techniques in more detail in the next section.

\subsection{Results and technical overview}

\paragraph*{Main results and outline of ideas.}
Our main results are summarized in \tblref{frechet_refs}. We present
(roughly) $O(cn/\eps)$ time algorithms for approximating the \Frechet
distance between two polygonal curves that are $c$-packed.

\begin{table}[h]
    \centering
    \begin{tabular}{l c c cc} %
      \toprule
      \textbf{Type}
      & \textbf{Discrete}
      &
      & \textbf{Continuous}
      & \\
      \midrule
      Weak
      & $\OO(\frac{cn}{\eps})$
      & \thmref{w_discrete_frechet}
      & $\OO(\frac{cn}{\eps})$
      & \thmref{w_cont_frechet}
      \\[0.1cm]
      Strong
      &  $\OO(\frac{cn}{\sqrt{\eps}} \log \tfrac{1}{\eps})$
      &  \corref{s_discrete_frechet_faster}
      &  $\OO(\frac{cn}{\sqrt{\eps}}\log^2 \tfrac{1}{\eps})$
      &
        \thmref{s_cont_frechet}
      \\[0.1cm]
      \bottomrule
    \end{tabular}
    \caption{Main results for $(1+\eps)$-approximation in this paper.  }
    \tbllab{frechet_refs}
\end{table}

Our new algorithms rely on several key insights: (I) the parametric space
can be approximated by a decomposition of linear size; (II) one can
propagate the computation of the \Frechet distance through this
non-uniform parametric space in linear time; and (III) this propagation
is relatively easy in the weak case, but in the strong (monotone) case,
one needs to carefully maintain the approximate distances computed so far
by maintaining an implicit, simplified representation, which introduces
significant technical challenges.

In particular, the monotonicity constraints, coupled with the
non-uniformity of the decomposition, require splitting the ``wavefront''
into several subparts and then merging some of them, even when handling a single cell.
Somewhat surprisingly, these operations can all be done in overall linear time.

This yields a linear-time constant-factor approximation algorithm. We can
then bootstrap these algorithms into a $(1+\eps)$-approximation by
performing a ``binary search'' over a constant-spread interval guaranteed
to contain the optimum, using an approximate decider. Because the spread
of this interval is constant, the overall running time remains linear for
a fixed $\eps$. By carefully implementing this search, we obtain an improved
dependence on $\eps$.

\bigskip

Ultimately, we obtain algorithms for a $(1+\eps)$-approximation of the
\Frechet distance that run in the same time required to decide whether
this distance lies in an interval of the form $[\alpha, (1+\eps)\alpha]$ for some real
number $\alpha$. Namely, our new algorithms avoid the additional $\log$
factors that almost always plague reductions of the optimization problem
to the decision problem. (We discuss this issue further in the
conclusions, see \secref{conclusions}.)

\paragraph*{Assumption.}

In the following, we assume that $\eps \in (0,1)$, the parameter determining the approximation quality, is some fixed constant independent of $n$. In particular, our results are less interesting in the regime $\eps \ll 1/\log^2 n$.  Similarly, while no assumption is made on the packedness parameter, conceptually, the results are more interesting if one considers it to be a constant. We assume that $n$ is the total complexity of the input curves (number of vertices).

\paragraph*{Preliminaries and setup.}
\secref{prelims} establishes basic notation and definitions, including
definitions of the main variants of the \Frechet distance. Each variant
is mapped to a bottleneck shortest path problem in a weighted domain
$\cPQ$. Known algorithms solve a variant of this problem in the form of a
path feasibility problem within a sublevel set of the weighted domain. In
contrast, the new algorithms solve a path-finding problem in the weighted
domain directly. To get an improved running time, we describe, in
\secref{faster}\footnote{This part is somewhat orthogonal to our main
   result and was moved towards the end of the paper for readability.},
how to turn a constant-factor approximation into a
$(1+\eps)$-approximation, while incurring only a constant-factor overhead
in the total running time, when comparing to the running time of the
decider.  In the main part of the paper, we focus on computing a
$(1+\eps)$-approximation in $\OO(\frac{cn}{\eps})$ time. We later use
this algorithm with $\eps=1/2$ to get a constant-factor approximation in
$\OO(cn)$ time that can then be turned into a $(1+\eps)$-approximation by
the algorithm from \secref{faster}.

\paragraph{The improved approximate decomposition.}
In \secref{decomposing-free-space}, we show how to decompose the domain $\cPQ$ into rectangles such that for almost all rectangles the weight within the rectangle is the same up to a factor of $(1+\eps)$. For the other rectangles, the weight function can still be $(1+\eps)$-approximated by a function that is simple enough for our purposes. The latter is only relevant for the continuous variants of the \Frechet distance. We refer to such functions below as \emph{simplified elevation functions}. The decomposition into rectangles is conceptually similar to a bipartite well-separated pair decomposition (WSPD). A similar data structure was used in the work~\cite{ypfa-seaad-16, afpy-adtwe-16} for approximating the Dynamic time warping distance. However, their algorithm and analysis were not sufficient for our purposes. Their bounds also include an extra logarithmic factor.  The size of the decomposition is $\OO(\frac{cn}{\eps})$ for two $c$-packed curves of $n$ vertices in total. The construction time, including computing all neighboring relations between rectangles, is in $\OO(\frac{cn}{\eps})$ (\corref{neighbors}).

\paragraph*{The basic algorithms.} We structure the presentation of our algorithms in two parts: \secref{frechet} and \secref{strong}. First, in \secref{frechet}, we develop approximation algorithms for the main variants of the \Frechet distance where no monotonicity is required.  Here, we first discuss the discrete \Frechet distance (\secref{weak-discrete}) and then the continuous \Frechet distance (\secref{weak-continuous}). This is relatively easy with the results of \secref{decomposing-free-space} at hand: In both of these settings, the computation of the \Frechet distances reduces to computing a bottleneck shortest path in a weighted undirected graph formed by the rectangles of the decomposition. In particular, the cost of an edge between two neighboring rectangles is determined by the minimum possible elevation at which the boundary between them can be crossed.  The resulting running time is in $\OO(\frac{cn}{\eps})$ for both discrete and continuous variants of the weak \Frechet distance (\thmref{w_discrete_frechet} and \thmref{w_cont_frechet}).  Note that these results do not require the framework of \secref{faster}.  Secondly, \secref{strong} presents our main result: an algorithm for the strong continuous \Frechet distance. The presentation is structured in several layers, each layer adding more complexity. In the following paragraphs, we discuss the challenges and techniques.

\paragraph*{Monotonicity.}
Handling monotonicity is a major challenge, as more information has to be
propagated across the domain. As mentioned above, \secref{strong}
describes the procedure for the (strong [i.e., monotone]) \Frechet
distance.  Our domain decomposition gives rise to a natural topological
order of the rectangles that respects monotonicity.  The algorithm
processes the rectangles (of the decomposition) in this topological
order. It maintains (approximately) the \Frechet distance between the
corresponding prefix curves for each point on the boundary between the
processed and the unprocessed rectangles, which is the \emph{propagation
   front}. Since the complexity of this front can be non-constant on the
boundary of a rectangle, even for the discrete \Frechet distance, lazy
updates are used to ensure constant amortized update time per processed
rectangle. For the (strong) discrete \Frechet distance, this results in a
$\OO(\frac{cn}{\eps})$ time $(1+\eps)$-approximation
algorithm~(\thmref{s_discrete_frechet}). The running time can be improved
to $\OO(\frac{cn}{\sqrt{\eps}}\textup{polylog} \tfrac{1}{\eps})$
(\corref{s_discrete_frechet_faster}) using the framework of
\secref{faster} in combination with the respective decision algorithm by
Bringmann and \Kunnemann~\cite{bk-iafdc-17}.

\paragraph*{Fragmentation.}
For the continuous \Frechet distance, the algorithm extends naturally using continuous extensions of discrete intervals, but there are yet other challenges. In this case, the weight of a rectangle of the decomposition is not necessarily approximated by a constant. Instead, more complex functions have to be used. In particular, an update to the cost function along the propagation front may lead to fragmentation (refer to \figref{Continuous_without_root2} for an example), such that it cannot easily be maintained in constant amortized time. \secref{strong-continuous} deals with tackling these challenges. We solve this by observing that the distance function between a point and an edge can be approximated nicely, avoiding the worst fragmentation behavior. We show that we can indeed compute a $(\sqrt{2}+\eps)$-approximation of the \Frechet distance in $\OO(\frac{cn}{\eps})$ time (\lemref{sqrt2_approx}). This can then be combined with the framework of \secref{faster} to obtain a $(1+\eps)$-approximation with improved running time. Here, again, we can plug in the respective decision algorithm by Bringmann and \Kunnemann~\cite{bk-iafdc-17} within the framework of \secref{faster} (\corref{s_discrete_frechet_faster}, \thmref{s_cont_frechet}).

This phenomenon that the distance function gets more complicated and fragments into many parts as it propagates through the domain is quite common. For example, this phenomenon is the main difficulty in getting efficient shortest-path algorithms on the boundary of convex polytopes in 3 dimensions \cite{mmp-dgp-87, ss-otasp-08}, and also in the shortest-path problem in the plane among obstacles \cite{hs-oaesp-99}. In the context of \Frechet distance, this fragmentation behavior was observed by \citeau{bblmm-cfdrl-16}. (However, the algorithm of \citeau{bblmm-cfdrl-16} runs in quadratic time, so their result does not seem to be directly relevant to our settings.)

\paragraph*{A postscript: A simpler algorithm for marching \Frechet.}
In addition to our main results, we present in \secref{marching} a simple approximation algorithm for the marching \Frechet distance, a variant of the discrete \Frechet distance that is, on the surface, very similar. This variant turns out to be much easier to approximate, and the algorithm can avoid the use of the data structure from \secref{decomposing-free-space} altogether. The running time of the algorithm is in $O(\frac{cn}{\eps})$ (\thmref{marching}). Note that this variant has been used in earlier work, e.g.,~\cite{avraham2015discrete}, but was not named in any particular way.

\section{Preliminaries}%
\seclab{prelims}

Let $\IRX{n}$ be the set of integers ${1,\ldots,n}$.  A continuous function $f: I \rightarrow \Re$ on an interval $I$ is a \emphi{valley function} if there exists an $x \in I$ such that $f$ is (not necessarily strictly) monotone decreasing on $I \cap \IOCX{-\infty, x}$ and (not necessarily strictly) monotone increasing on $I \cap \ICOX{x, \infty}$.

Let $P=p_1,\ldots,p_n$ be a sequence of points (i.e., \emphw{vertices}) in $\Re^d$ (for simplicity, we assume no two consecutive points in these sequences are equal), where $n = \cardin{P}$. The polygonal curve associated with $P$, denoted by $\cP$, is formed by connecting consecutive vertices of $P$ by segments. The \emphw{length} and \emphi{domain} of $\cP$ are, respectively,
\begin{equation*}
    L
    =%
    \lenX{\cP}=\sum_{i=1}^{n-1}\dY{p_{i+1}}{p_i}
    \qquad\text{and}\qquad%
    \cXP =  [0,L].
\end{equation*}
The associated function $\cP\colon\cXP\rightarrow\Re^d$ uniformly parameterizes this curve, where $\cP(0) = p_1$, $\cP(L) = p_n$, and $\cP(t)$ is the point in distance $t$ from $p_1$ along $\cP$, for any $t \in \cXP$. In particular, for $i \leq j$, let $\cPY{i}{j}$ denote the subcurve of $\cP$ from $p_i$ to $p_j$.  For a vertex $p_i \in P$, its \emphi{curve coordinate} is $x_i = \lenX{\cPY{1}{i}} = \sum_{j=1}^{i-1}\dY{p_j}{p_{j+1}}$, and $\XP = ( x_1, \ldots, x_{n})$ is the \emphi{coordinate sequence} of $P$.  Similarly, for the other input curve $Q = q_1, \ldots, q_m$, let $\cYQ = [0,\lenX{\contX{Q}}]$ and its corresponding coordinate sequence is $\YQ = (y_1, \ldots, y_m)$, with $y_j = \lenX{\cQY{1}{j}}$, for all $j$.

\begin{defn}
    \deflab{c_packed}%
    For $c > 0$, a polygonal curve $\cP$ is \emphw{$c$-packed} \cite{dhw-afdrc-12}, if for all $p \in \Re^d$, and $r > 0$, we have that $\lenX{ \cP \cap \ballY{p}{r}} \leq c r$, where $\ballY{p}{r}$ denotes the ball of radius $r$ centered at~$p$. A sequence~$P$ is \emphi{$c$-packed} if $\cP$ is $c$-packed.
\end{defn}

\subsection{\Frechet distance definitions}%

\subsubsection{The continuous \Frechet distance}
\seclab{continuous}

Let $\cPQ = [0,\lenX{P}] \times [0, \lenX{Q}]$ denote the \emphw{parametric space} associated with $P$ and $Q$. The \emphi{elevation} is the function $\elev$, defined over $\cPQ$, where $\elev(x,y)=\dY{\cP(x)}{\cQ(y)}$.  An \emphi{alignment} is a curve $\alignment \subseteq \cPQ$ that starts at $(0,0)$ and ends at $(\lenX{P},\lenX{Q})$. An alignment is \emphw{monotone} if $\alignment$ is $x$-monotone and $y$-monotone.  The \emphi{pass} of an alignment $\alignment$ is
\begin{equation}
    \passX{\alignment}
    =%
    \max_{(x,y) \in \alignment }\elev\bigl(x,y \bigr),
    \eqlab{pass}
\end{equation}
which is the maximum elevation that $\alignment$ attains anywhere.  The (strong continuous) \emphi{\Frechet distance} between $P$ and $Q$ is the minimum pass of any monotone alignment between $P$ and $Q$. Formally, it is
\[
    \dFY{\cP}{\cQ}%
    =%
    \min_{\alignment \in {\mathcal{A}}} \passX{\alignment},
\]
where ${\mathcal{A}}$ denotes the set of all monotone alignments in $\cPQ$.  The \emphi{weak} variant of this distance removes the requirement for monotonicity, minimizing over the set of all alignments.

\begin{remark}
    The above formulation is an alternative way of stating the standard definition \cite{ag-cfdbt-95}.  The formalization of the \Frechet distance is via reparametrizations of the curves that are required to be bijections~\cite{dd-ed-09}. However, to enable one of the entities to stay in place in the continuous case, one has to consider the optimal reparameterization as the limit of such bijections. The above formalization avoids the technicality by directly allowing portions of the alignment to be horizontal or vertical. For a formal proof that the definitions are equivalent, see e.g.~\cite{bdr-aklcp-23}.
\end{remark}

\begin{defn}
    \deflab{cell}%
    A subset of the domain of the form $I\times J$ with intervals $I \subseteq \cXP$, $J \subseteq \cYQ$ is a \emphi{rectangle}. A \emphi{cell} of the elevation function is a rectangle that corresponds to one edge of one curve and an edge of the other curve.  It is known that the sublevel set of the elevation function inside a cell is a (clipped) ellipse \cite{ag-cfdbt-95}. It is not hard to show that the elevation function is also convex.
\end{defn}

\subsubsection{The discrete \Frechet distance}%
\seclab{discrete_f}

A \emphi{(discrete) interval} $J = ( x_{i}, x_{i+1}, \ldots, x_j )$ of $P$ is a consecutive subsequence of $\XP$. It will be convenient to interpret such discrete intervals as being associated with continuous intervals.

\begin{defn}
    \deflab{extension}%
    The \emphi{extension} of a discrete interval $J = ( x_{i}, x_{i+1}, \ldots, x_j )$ is the interval $\suppX{J} = \IOCX{x_{i-1}, x_j}$ with $x_0 = 0^-$. For a continuous interval $J$, we have $\suppX{J} = J$.
\end{defn}

In the discrete setting, the \emphw{parametric space} associated with $P$ and $Q$ is $\PQ = \XP \times \YQ$.  The set $\PQ$ is an embedding of the natural grid $\IRX{\cardin{P}} \times \IRX{\cardin{Q}}$ so that the pair of vertices $(p_i,q_j)$ (i.e., $(i,j)$) are now located at their respective curve coordinates $(x_i,y_j)$.  For two discrete intervals $I$ and $J$, of $\XP$ and $\YQ$, respectively, their associated (discrete) \emphi{rectangle} is $\rect = I \times J \subseteq \PQ$. The \emphi{extension} of~$\rect$ is $\suppX{\rect} = \suppX{I} \times \suppX{J}$. For a continuous rectangle $\rect$, we have $\suppX{\rect} = \rect$.

A \emphi{discrete alignment} $\alignment$ between sequences $P$ and $Q$ of $n$ and $m$ vertices, respectively, is a sequence of pairs $(i_1, j_1), \dots, (i_k, j_k)$ where $(i_1, j_1)=(1, 1)$, $(i_k, j_k)=(n,m)$, and $\dY{(i_{\ell}, j_{\ell})}{(i_{\ell+1}, j_{\ell+1})}_\infty=1$ for all $\ell$. A discrete alignment is \emphw{monotone} if $i_\ell\leq i_{\ell+1}$ and $j_\ell\leq j_{\ell+1}$, for all $\ell$.  The \emphw{pass} of $\alignment$, denoted by $\passX{\alignment}$ is its maximum elevation.  The (strong) \emphi{discrete \Frechet distance} between $P$ and~$Q$ is
\begin{math}
    \dFdY{P}{Q} =%
    \min_{\alignment}\passX{\alignment},
\end{math}
where the minimum is over all monotone discrete alignments.  The \emphw{weak} variant minimizes over the set of all discrete alignments, instead of monotone alignments.

\begin{remark}
    A \emphw{marching alignment} requires that $\dY{(i_{\ell}, j_{\ell})}{(i_{\ell+1}, j_{\ell+1})}_1=1$ (namely, only horizontal or vertical moves are allowed in the parametric space).  The \emphw{marching \Frechet distance} is the resulting distance when minimizing over marching alignments. In contrast, the (default) discrete \Frechet distance also allows diagonal moves. The distance under the marching version might be (significantly) larger than the (default) discrete \Frechet distance, but it is easier to approximate. See \secref{marching} for details.
\end{remark}

Interpreting the \Frechet distance as a graph problem (in the discrete case) is a bottleneck shortest path problem, where the weights are on the vertices.

\section{Approximating the elevation function}
\seclab{decomposing-free-space}

In this section, we describe how to decompose the domain of the elevation function into rectangles wherein the distance function can be approximated nicely.

Given a sequence $P = p_1, \ldots, p_n $ of points, build a binary tree bottom up on this sequence, where the leaves are the points (stored in their sequence order in the tree), and every internal node has two children. For every point in the sequence, create a singleton tree containing this point.  Then, having a list of such trees $T_1, \ldots, T_n $, the algorithm repeatedly scans this list, creating a higher tree by merging two adjacent trees in this list (leaving potentially the last tree in the list as is if the list has odd length), and repeating the process till there is a single tree left, denoted by $\Tp$.  The subsequence of $P$ stored in the subtree of $\vp$ is denoted by $\VX{\vp} = (p_i, \ldots, p_j)$, and its corresponding interval is $\InodeX{\vp} = (x_i, \ldots, x_j)$. The corresponding continuous subcurve is $\cVX{\vp} = \cP\!\IOCX{x_{i-1}, x_j}$, where $\cP\!\IOCX{x_{i-1},x_j}$ denotes the curve $\cP$ restricted to the interval $\IOCX{x_{i-1}, x_j} = \suppX{\bigl. \InodeX{\vp}}$.

\subsection{The information stored in the tree}

We store for each node~$\vp \in \Tp$, the following information:
\begin{compactenumA}
    \item $\InodeX{\vp} \subseteq \XP$, $\suppX{\vp} = \suppX{\InodeX{\vp}} \subseteq \cXP$ (all defined above).
    \item $\VX{\vp}$, $\cVX{\vp}$: The subcurves (also defined above). These are stored implicitly by (A).
    \item $\dmX{\vp} = \lenX{\cVX{\vp}}$: The total length of the subcurve.

    \item $\repX{\vp}$: The last vertex of the subcurve, it serves as its \emphw{representative}.

    \item $\pX{\vp}$: The parent node of $\vp$ in $\Tp$.
    \item Its two children, if they exist.
\end{compactenumA}

Observe that if $\vp$ is a leaf, then $\cVX{\vp}$ is a segment. In this case, it is denoted by $\segX{\vp}$. By storing the relevant indices, the above information requires $O(1)$ space in the node.

\begin{observation}
    Computing the two trees $\Tp$ and $\Tq$ with all the information takes $\OO(n)$ time, where $n$ is the total number of vertices of $P$ and $Q$.
\end{observation}

The above algorithm computes two trees $\Tp$ and $\Tq$ for the two input curves. The next step is to use these trees to compute a decomposition of the domain $\cPQ = \cXP \times \cYQ$.%

\begin{defn}
    \deflab{dist_sets}%
    For two non-empty sets $X,Y \subseteq \Re^d$, let $\dsY{X}{Y} = \inf_{x \in X, y \in Y} \dY{x}{y}$. If $X = \{x\}$ is a singleton, we abuse notations and write $\dsY{x}{Y} = \dsY{X}{Y}$.
\end{defn}

\begin{defn}
    \deflab{separated}%
    A pair of nodes $(\vp, \vq)$ with $\vp \in \Tp$ and $\vq \in \Tq$ is \emphi{$\eps$-separated}, for $\eps \in (0,1)$, if one of the following holds:
    \begin{compactenumI}
        \item \itemlab{W_I} Both $\vp$ and $\vq$ are leaves.

        \item \itemlab{W_II} $\vp$ is a leaf and $\dsY{\segX{\vp}}{\repX{\vq}} \geq \frac{4}{\eps} \dmX{\vq}$.

        \item \itemlab{W_III} $\vq$ is a leaf and $\dsY{\segX{\vq}}{\repX{\vp}} \geq \frac{4}{\eps} \dmX{\vp}$.

        \item \itemlab{W_IV} $\dY{\repX{\vp}}{\repX{\vq}} \geq \frac{4}{\eps} \max\bigl( \dmX{\vp}, \dmX{\vq}\bigr)$.
    \end{compactenumI}
\end{defn}

The algorithm uses the trees $\Tp$ and $\Tq$ to decompose the domain $\cPQ$ into rectangles by using an axis-aligned binary space partition. The idea is to emulate the \WSPD construction algorithm of Callahan and Kosaraju \cite{ck-dmpsa-95} by constructing a tree $\Tbsp$. Each node of this tree is a pair of nodes in $\Tp\times\Tq$, and each leaf $(\vp,\vq) \in \Tbsp$ is a pair that is $\eps$-separated.  The rectangle associated with such a node~$(\vp,\vq)$ is $\rectY{\vp}{\vq} = \suppX{\vp} \times \suppX{\vq}$. Thus, the leaves of $\Tbsp$ induce a decomposition of the set $\cPQ$ into these rectangles.  Initially, the root of $\Tbsp$ is $(\rootX{\Tp}, \rootX{\Tq})$.  The refinement algorithm maintains a (FIFO) queue of (some of the current) leaves of $\Tbsp $ that have not yet been handled.  The algorithm fetches the first pair $u=(\vp, \vq)$ from the queue.

The algorithm applies the first option that fits:
\begin{compactenumA}
    \item If the pair $u$ is $\eps$-separated, then this pair is \emphw{final} (this includes the case that both $\vp$ and $\vq$ are leaves), and the algorithm moves on to the next pair in the queue. As a result, $u$ is a leaf of the final tree~$\Tbsp$.

    \item $\dmX{\vp} \leq \dmX{\vq}$ and $\vq$ is not a leaf: The algorithm replaces $\vq$ by its two children $\vq_1$ and $\vq_2$, by pushing into the queue the newly created pairs $u_1 = (\vp, \vq_1)$ and $u_2 = (\vp,\vq_2)$. Next, it creates two new leaves $u_1$ and $u_2$, and inserts them as children of $u$ in $\Tbsp$.

    \item $\dmX{\vp} > \dmX{\vq}$ and $\vp$ is not a leaf: Same as the above, with the two children of $\vp$ being used.%

    \item \itemlab{split_the_other} Either $\vp$ or $\vq$ is a leaf (but not both). The algorithm splits the pair using the children of the node that is not a leaf and adds them to the queue.
\end{compactenumA}
The algorithm then continues to the next pair in the queue. The algorithm continues this process till the queue is empty.

\bigskip

The output of this algorithm is the \emphi{decomposition} $\cDecomp=\Set{\rectY{\vp}{\vq}}{ (\vp,\vq) \text{ leaf of }\Tcont }$, which provides a disjoint cover of the domain $\cPQ = [0, \lenX{P}] \times [ 0, \lenX{Q}]$.

\subsection{Analysis}

\begin{lemma}
    \lemlab{wspd}%
    We have the following:
    \begin{compactenumi}
        \item For any $(x,y) \in \cPQ$, and the corresponding points $p = \cP(x)$ and $q = \cQ(y)$, there is a unique leaf $(\vp,\vq)\in \Tbsp$, such that $(x,y ) \in \rect_{(\vp,\vq)}$.

        \item If $(\vp,\vq)$ is a final pair, then for any pair of vertices $p\in \VX{\vp}$ and $q\in \VX{\vq}$, we have that
        \[
            (1-\eps)\dY{p}{q}\leq\dY{\repX{\vp}}{\repX{\vq}} \leq (1+\eps)\dY{p}{q}.
        \]
    \end{compactenumi}
\end{lemma}

\begin{proof}
    (i) This follows by induction on the construction algorithm, tracking the unique current leaf of $\Tbsp$ whose rectangle contains $(x,y)$.

    (ii) Assume that $(\vp,\vq)$ is a final pair and $p\in\vp$ and $q\in\vq$ are vertices. If this pair was finalized by condition \itemref{W_IV} (when the algorithm checks if the pair is $\eps$-separated), this implies, for $D = \max\bigl( \dmX{\vp}, \dmX{\vq}\bigr)$, that $D \leq \frac{\eps}{4}\dY{\repX{\vp}}{\repX{\vq}}$. Thus, by the triangle inequality and $\dmX{\vp} \leq D$ and $\dmX{\vq} \leq D$, we have
    \begin{equation}
        (1-\tfrac{\eps}{2})\dY{\repX{\vp}}{\repX{\vq}}
        \leq
        \dY{\repX{\vp}}{\repX{\vq}}-2D
        \leq
        \dY{p}{q}
        \leq
        \dY{\repX{\vp}}{\repX{\vq}}+2D
        \leq
        (1+\tfrac{\eps}{2})\dY{\repX{\vp}}{\repX{\vq}}.
        \eqlab{yo}
    \end{equation}
    Using $\tfrac{1}{1-\eps/2}\leq(1+\eps)$ and $\tfrac{1}{1+\eps/2}\geq(1-\eps)$ for $\eps\in(0,1)$, this yields the desired inequality.

    If the pair is finalized by condition \itemref{W_II}, then $\vp$ is a leaf, $\repX{\vp} = p$, and $\dY{\repX{\vp}}{\repX{\vq}} \geq \dsY{\segX{\vp}}{\repX{\vq}} \geq \frac{4}{\eps} \dmX{\vq}$. Setting $D = \dmX{\vq}/2$, and for all $q \in \VX{\vq}$, \Eqref{yo} applies and implies the claim.  The same applies if the pair is finalized by condition \itemref{W_III}.

    The only remaining case is that both nodes are leaves, but then $p = \repX{\vp}$ and $q = \repX{\vq}$, and the claim trivially holds.
\end{proof}

\begin{remark}
    \remlab{wspd_like}
    \begin{compactenumi*}%
        \item \lemref{wspd} (ii) implies that the decomposition
        \begin{math}
            \mathcal{W} = \Set{ \{ \VX{\vp}, \VX{\vq} \} }{ (\vp,\vq) \text{ leaf of }\Tbsp}
        \end{math}
        is a bichromatic $1/\eps$-\WSPD of $P$ and $Q$. Below we show that this decomposition is of size $O( cn/\eps)$, which is better than the regular bounds on such decompositions (i.e., $O(n /\eps^d)$). However, this improved bound is only for the case that $P$ and $Q$ are $c$-packed. In particular, unlike standard \WSPD constructions, this decomposition can have quadratic size in the worst case (e.g., for $c = \Omega(n)$).

        \indent%
        \item \itemlab{its_a_constant} %
        Another implication of \lemref{wspd} (ii), for a rectangle $\rect$ defined by a leaf of $\Tbsp$, is that the elevation function over $\rect \cap ( \XP \times \YQ)$ is a constant up to a factor of $1 \pm \eps$. That is, conceptually, on the discrete domain, the elevation function is (approximately) a constant inside each rectangle.
    \end{compactenumi*}
\end{remark}

A node $(\vp,\vq) \in \Tbsp$ that was created by the algorithm at step \itemref{split_the_other} in the refinement is a \emphi{minor} node. The children of a minor node are also minor. Our argument will first bound the number of non-minor nodes. Before, we prove a simple helper lemma. In the following, we denote by $\parentX{\vp}$ the parent of a node $\vp$.

\begin{lemma} \lemlab{stolen} Consider a pair $(\vp,\vq) \in \Tbsp$:
    \begin{compactenumi}
        \item If $(\vp,\vq)$ is non-minor, then $\min \bigl( \dmX{\parentX{\vp}}, \dmX{\parentX{\vq}} \bigr)\geq \max\bigl( \dmX{\vp}, \dmX{\vq} \bigr)$, and

        \item If $(\vp,\vq)$ is minor and $\vq$ is an internal node of $\Tq$, then $\dmX{\parentX{\vp}} \geq \dmX{\vq}$.
    \end{compactenumi}
\end{lemma}
\begin{proof}
    (i) Consider two nodes $(\vp_1,\vq_1), (\vp_2,\vq_2) \in \Tbsp$ that are adjacent, with say, $\vp_1 = \vp_2$ and $\vq_1 = \pX{\vq_2}$.  As the algorithm always splits the larger side of non-minor pairs, and the diameters of nodes are monotonically decreasing along paths in $\Tp$ and $\Tq$ (as one goes down), we have that $\max\bigl( \dmX{\vp_2}, \dmX{\vq_2} \bigr) \leq \dmX{\vq_1}$. Similarly, if $\vp_1 = \pX{\vp_2}$ and $\vq_1 = \vq_2$, then $\max\bigl( \dmX{\vp_2}, \dmX{\vq_2} \bigr) \leq \dmX{\vp_1}$. The claim now readily follows by induction, by observing that $\pX{\vp}$ and $\pX{\vq}$ both appear somewhere in the pairs (i.e., nodes) of the path from $(\vp,\vq)$ to the root of $\Tbsp$, see \figref{stolen}.

    \begin{figure}[h]
        \centerline{%
           \includegraphics{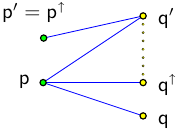}%
        }%
        \caption{}
        \figlab{stolen}
    \end{figure}

    (ii) Let $(\vp_0,\vq_0) \in \Tbsp$ be the lowest ancestor of $(\vp,\vq)$ that is not minor. By construction, $\vp$ is a leaf of $\Tp$, and $\vp_0 = \vp$. By (i) we have that $\dmX{\pX{\vp}} \geq \dmX{\vq_0} \geq \dmX{\vq}$.
\end{proof}

\begin{lemma}
    \lemlab{num_pairs}%
    Let $n$ be the total number of vertices of $P$ and $Q$.  If $\cP$ and $\cQ$ are $c$-packed, then the tree $\Tbsp$ constructed above contains $\OO(\frac{cn}{\eps})$ non-minor nodes.
\end{lemma}
\begin{proof}
    Let $X$ be the set of all nodes in $\Tbsp$ that are not minor.  Let $Y$ be the set of all nodes $v \in X$ that are either a leaf or whose children are minor in $\Tbsp$. Clearly $|X| = O(|Y|)$. So, bounding $|Y|$ is sufficient to bound the number of non-minor nodes in $\Tbsp$. In the following, we describe a charging scheme to charge the nodes in $Y$ to nodes in the trees $\Tp$ and $\Tq$. We then argue by using a packing argument that the number of pairs (nodes of $\Tbsp$) charged to a node in $\Tp$ or $\Tq$ is small.

    Assume that $(\vp, \parentX{\vq})\in X$ (the other case is $(\parentX{\vp}, \vq)\in X$, which is symmetric by switching the role of $\vp$ and $\vq$).  We charge the pair $(\vp,\vq)$ to $\parentX{\vq}$. Let $C$ be the set of all the vertices $x \in \Tp$ such that $(x, \vq) \in Y$ was charged to $\parentX{\vq}$.  Note that at most $2\cardin{C}$ charges are made to $\parentX{\vq}$ because for $(x,\vq)$ to be charged to $\parentX{\vq}$, $\vq$ has to be one of the two children of $\parentX{\vq}$. Observe that no pair of vertices $x_1, x_2 \in C$ has the property that $x_1$ is an ancestor of $ x_2$. Indeed, if $(x_1, \vq) \in Y$, then its children are minor, and the pair $(x_2,\vq)$ is minor as well, and is thus not in $Y$. In addition, for any $x \in C$ we have that $(x,\pX{\vq}) \in X \setminus Y$.  Therefore the pair is not $\eps$-separated and thus by \defref{separated} and \lemref{stolen}, for any $x \in C$, we have that
    \begin{equation*}
        \dY{\repX{x}}{\repX{\parentX{\vq}}}
        \leq
        \tfrac{4}{\eps}  \max\bigl( \dmX{x}, \dmX{\smash{\pX{\vq}}} \bigr)
        =
        \tfrac{4}{\eps} \dmX{\pX{\vq}}.
    \end{equation*}
    Let $F = \Set{ \cVX{\!\parentX{\,x}}\!}{ x \in C}$.  Note that for two distinct $x_1, x_2\in C$, we either have $\parentX{x_1}=\parentX{x_2}$ or $\cVX{\parentX{x_1}}$ and $\cVX{\parentX{x_2}}$ are disjoint, since $(x, \parentX{\vq})\in X$. So, the set $F$ contains only disjoint subcurves of $\cP$, and $\cardin{F}\geq\cardin{C}/2$. All of the curves in $F$ contain points within distance at most $R=\frac{4}{\eps}\dmX{\pX{\vq}}$ from $\repX{\pX{\,\vq}}$.  Each one of these fragments have length at least $\dmX{\pX{\vq}}$ because $\dmX{\pX{x}}\geq\dmX{\pX{\vq}}$ for all $x\in C$ by the proof of \lemref{stolen}.  Note that each of these fragments either fully lies inside $\ballY{\repX{\vq}}{2 R}$ or is longer than $R>\dmX{\pX{\vq}}$. As such, each fragment contributes at least a length of $\dmX{\pX{\vq}}$ to the length of the curve inside the ball.  But then, as~$\cP$ is \defrefY{c_packed}{$c$-packed} in the ball $\ballY{\repX{\vq}}{2R}$, we have that $\cardin{F}\cdot \dmX{\pX{\vq}} \leq\lenX{\cP\cap\ballY{\repX{\vq}}{2R}}\leq2cR$ implying
    \begin{equation*}
        \cardin{F}
        \leq
        2c R/\dmX{\pX{\vq}}
        =
        8c/\eps.
    \end{equation*}
    This implies $|C| = O(c/\eps)$, which bounds the number of charges $\pX{\vq}$ pays for. As~$\Tp$ has size $\OO(n)$, it follows that $|Y| = O( cn/\eps)$, and this bounds the number of non-minor pairs considered by the algorithm.
\end{proof}

\begin{lemma}
    \lemlab{cont_time_decomp}%
    Let $n$ be the total number of vertices of $P$ and $Q$.  If $\cP$ and $\cQ$ are $c$-packed, the tree~$\Tbsp$ contains at most $\OO( \frac{cn}{\eps} )$ nodes.
\end{lemma}
\begin{proof}
    For simplicity of exposition, assume that $\Tp$ has all its leaves at the same distance from its root\footnote{One could make the number of vertices in $P$ a power of two by adding fake vertices, and then this property holds.}.  Let $U$ be the set of all minor nodes in $\Tbsp$ that are not leaves. The task is to bound $\cardin{U}$ as all minor leaves can be charged directly to nodes of $U$.

    Fix an internal node $\vq \in \Tq$, and let $F_\vq = \Set{ \vp \in \Tp }{ (\vp,\vq) \in U}$ and $C_\vq = \Set{ \cVX{\pX{\vp}}\!}{ \vp \in F_{\vq} }$. All the nodes in $F_\vq$ are leaves of $\Tp$. This, together with the fact that $\Tp$ is a binary tree, implies that at most two nodes of $F_\vq$ might be mapped to the same fragment of $C_\vq$. The fragments of $C_\vq$ are interior disjoint and are relatively long.  Indeed, for each pair $(\vp,\vq)\in U$, we have $\dmX{\pX{\vp}} \geq \dmX{\vq}$ by \lemref{stolen}. Similarly, since $(\vp,\vq)\in U$ is not final by definition, it must be that the separation check (see \defref{separated}) failed (i.e., \itemref{W_II} does not hold). This implies that
    \begin{equation*}
        \dsY{\cVX{\pX{\vp}}}{\repX{\vq}} %
        \leq
        \dsY{\segX{\vp}}{\repX{\vq}} < \frac{4}{\eps} \dmX{\vq}.
    \end{equation*}

    So, the fragments of $C_\vq$ intersect the ball $\ballY{\repX{\vq}}{r}$, and thus the enlarged ball $\ballC = \ballY{\repX{\vq}}{2r}$, where $r = {\tfrac{4}{\eps} \dmX{\vq}}$.  By the $c$-packed property of $\cP$, the total length of $\cP$ inside $\ballC$ is at most $2cr$, which readily implies that $\cardin{C_\vq} \leq 8c/\eps$. Thus, any internal node of $\Tp$ or $\Tq$ pays for at most $O(c/\eps)$ minor nodes, which implies the desired bound by \lemref{num_pairs}, as $\cardin{U} = \sum_{\vq \in \Tq}O( \cardin{C_\vq})$.
\end{proof}

The above implies the following result.

\begin{theorem}
    \thmlab{number_pairs}
    Given two $c$-packed polygonal curves $\cP$ and $\cQ$, with a total of $n$ vertices, and a parameter $\eps \in (0,1)$, the above algorithm computes, in $\OO(cn/ \eps)$ time, a decomposition of $\cPQ = [0,\lenX{P}] \times [0, \lenX{Q}]$ into $\OO(cn/\eps)$ disjoint rectangles.  A rectangle $\rectY{\vp}{\vq}$ in this decomposition is induced by a pair of nodes~$\vp \in \Tp$ and~$\vq \in \Tq$, where $\Tp$ and $\Tq$ are two binary trees (computed by the algorithm) of linear size on~${\cP}$ and ${\cQ}$, respectively, such that $\vp$ and $\vq$ are \defrefY{separated}{$\eps$-separated}.
\end{theorem}

\begin{remark}
    (i) The only properties of $\Tp$ (and $\Tq$) used in the above construction are that they have linear size in the input curves, and they form a hierarchical decomposition of the underlying curve~$\cP$. Thus, what exact binary tree to use in the above construction is somewhat arbitrary. We choose the simplest trees that work for our purposes.

    (ii) One can view the above algorithm as computing a bichromatic \eWSPD of~$P \otimes Q$, similar to the data structure in~\citeau{ypfa-seaad-16} (and the follow-up work \cite{afpy-adtwe-16}).
\end{remark}

\subsection{Using the decomposition}
\seclab{refinement}

\begin{defn}
    \deflab{simp_elev}%
    For a rectangle $\rect=\rectY{\vp}{\vq}\subseteq\cPQ$, a \emphi{simplified elevation function} $\elevS:\rect \rightarrow \Re$, is a function of the form $\esX{x,y} = \normX{x u + y v + w }$ where $u,v, w$ are arbitrary vectors in $\Re^d$ (which might be $0$).
\end{defn}

If $\rect$ is a \defrefY{cell}{cell}, then the original elevation function is $\elev(x,y) = \normX{ \cP(x)-\cQ(y)}$, and it can be interpreted as being a simplified elevation function.  The main useful property of the decomposition of \thmref{number_pairs} is that the elevation function, for every rectangle in it, can be well approximated by a simplified elevation function.

\begin{lemma}
    \lemlab{cont_fs_approx}%
    For each rectangle of the decomposition of \thmref{number_pairs}, there is an associated \defrefY{simp_elev}{simplified elevation} function $\elevS:\rect\rightarrow\Re$, such that for all $(x,y) \in \rect$,
    \begin{equation*}
        (1-\eps)\elevY{x}{y} \leq \esX{x,y} \leq (1+\eps) \elevY{x}{y}.
    \end{equation*}
\end{lemma}
\begin{proof}
    Consider a leaf~$(\vp,\vq)$ of the tree $\Tbsp $ and its associated rectangle $\rect = \rectY{\vp}{\vq}=\suppX{\vp} \times \suppX{\vq}$ in the decomposition.  By construction, $(\vp,\vq)$ is $\eps$-separated. Namely, $(\vp,\vq)$ fulfills one of the conditions of \defref{separated}.

    If $(\vp,\vq)$ was finalized because of condition \itemref{W_I}, then $\rect$ is a cell, and its elevation function is already in its simplified form.

    If $(\vp,\vq)$ was finalized because of condition \itemref{W_II}, then $\vp$ is a leaf and $\dsY{\segX{\vp}}{\repX{\vq}} \geq \frac{4}{\eps} \dmX{\vq}$. In this case, let
    \begin{equation*}
        \esX{x,y}
        =
        \dY{\cP(x)}{\repX{\vq}},
    \end{equation*}
    and observe that $\cP(x)$ is an affine function (as it is moving linearly along its respective edge), and thus $\esX{x,y}$ is indeed simplified. As for the approximation, by the triangle inequality, we have that
    \begin{equation*}
        \nabla
        =
        \cardin{\esX{x,y} - \elev(x,y)}
        =%
        \cardin{\Bigl. \dY{\cP(x)}{\repX{\vq}} -
           \dY{\cP(x)}{\cQ(y)} }
        \leq
        \dY{\repX{\vq}}{\cQ(y)}
        \leq
        \dmX{\vq}
        \leq
        \frac{\eps}{4}
        \dsY{\segX{\vp}}{\repX{\vq}}.
    \end{equation*}
    However,
    \begin{math}
        \elev(x,y) \geq \dsY{\segX{\vp}}{\repX{\vq}} - \dmX{\vq} \geq \dsY{\segX{\vp}}{\repX{\vq}}/2.
    \end{math}
    This implies that $\nabla \leq (\eps/2) \elev(x,y)$, as desired.

    \smallskip%
    The case \itemref{W_III} is similar, with the simplified function being
    \begin{math}
        \esX{x,y} = \dY{\repX{\vp}}{\cQ(y)},
    \end{math}

    Finally, if \itemref{W_IV} finalized the pair, then $\dY{\repX{\vp}}{\repX{\vq}} \geq \frac{4}{\eps} \max\bigl( \dmX{\vp}, \dmX{\vq}\bigr)$.  Setting $\esX{x,y} = \dY{\repX{\vp}}{\repX{\vq}}$, we have that
    \begin{equation*}
        \elev(x,y)
        \geq
        \dY{\repX{\vp}}{\repX{\vq}}
        - \dmX{\vp} - \dmX{\vq}
        \geq
        (1-\tfrac{\eps}{2}) \dY{\repX{\vp}}{\repX{\vq}}
        =
        (1-\tfrac{\eps}{2})\esX{x,y}
        .
    \end{equation*}
    And similarly, $\elev(x,y) \leq (1+\tfrac{\eps}{2}) \esX{x,y}$.  The claim now readily follows by observing that $\frac{1}{1+\eps/2} \geq 1-\eps$ and $\frac{1}{1-\eps/2} \leq 1+\eps $.
\end{proof}

\begin{remark}
    The proof of \lemref{cont_fs_approx} provides a characterization of the simplified elevation function for a rectangle $\rect_{\vp,\vq} \in \cDecomp$ and any point $(x,y) \in \rect_{\vp,\vq}$:
    \begin{align*}
        \esX{x,y} =
        \begin{cases}
          \dY{\repX{\vp}}{\repX{\vq}}
          & \text{if }
            \dY{\repX{\vp}}{\repX{\vq}} \geq \frac{4}{\eps} \max\bigl(
            \dmX{\vp}, \dmX{\vq}\bigr)\text{,}
          \\
          \dY{\cP(x)}{\repX{\vq}}
          & \text{if } \vp \text{ is a leaf and } \vq \text{ is not a leaf,}\\
          \dY{\repX{\vp}}{\cQ(y)}
          & \text{if } \vq \text{ is a leaf and } \vp \text{ is not a leaf,}  \\
          \dY{\cP(x)}{\cQ(y)}
          & \text{otherwise (in this case $\rect_{\vp,\vq}$ is a cell)}.
        \end{cases}
    \end{align*}

    Note that for any pair of vertices $p_i\in P$, $q_j\in Q$ with curve coordinates $x_i\in\XP$, $y_j\in\YQ$ such that $(x_i,y_j)\in\rectY{\vp}{\vq}\in\cDecomp$, we have $\esX{x_i,y_j}=\dY{\repX{\vp}}{\repX{\vq}}$. So in the discrete case, we can consider the simplified elevation function to be constant on each rectangle.
\end{remark}

\subsection{Computing neighboring cells}

The associated discrete domain, for the sequences $P$ and $Q$, is the grid $\PQ = \XP \times \YQ$, where $\XP =\{x_1, \ldots, x_{\cardin{P}}\}$ and $\YQ = \{y_1, \ldots, y_{\cardin{Q}}\}$ are the coordinate sequences of $P$ and $Q$, respectively. A neat property of $\cDecomp$ is that every rectangle in it contains at least one point of $\PQ$. Thus, the associated (discrete) decomposition is
\begin{math}
    \Decomp = \Set{ \rect \cap \PQ}{\rect \in \cDecomp}.
\end{math}
Note that for algorithmic purposes, the two decompositions are essentially the same.

\begin{defn}
    Two rectangles $\rect, \rect' \in \cDecomp$ are \emphi{neighbors} if $\closureX{\rect} \cap \closureX{\rect'} \neq \emptyset$, where $\closureX{\rect}$ denotes the closure of $\rect$.  Two discrete rectangles $\rect,\rect'\in \Decomp$ are \emphi{neighbors} if their extensions are neighbors. Let $\Neighbors{\rect}$ denote the set of all the neighbors of a rectangle $\rect$.
\end{defn}

\begin{observation}
    \obslab{laminar}%
    The tree $\Tp$ (and also $\Tq$) defines a \emphi{laminar} family of intervals. That is, for any $\vp, \vp' \in \Tp$, we have that either $\suppX{\vp}$ and $\suppX{\vp'}$ are disjoint, or one of them contains the other. Any rectangle $\rect_u$ associated with a node $u \in \Tcont$ is the product of two intervals from these two laminar families.
    Thus, for any two rectangles $\rect_1, \rect_2 \in \cDecomp$, exactly one of the following is true
    \begin{compactenumi}
        \item $\closureX{\rect_1} \subseteq \closureX{\rect_2}$ or $\closureX{\rect_2} \subseteq \closureX{\rect_1}$,
        \item $\closureX{\rect_1}$ and $\closureX{\rect_2}$ are disjoint,

        \item $\closureX{\rect_1} \cap \closureX{\rect_2}$ is a segment, that is the side edge of one of the rectangles $\rect_1$ and $\rect_2$.

        \item $\closureX{\rect_1} \cap \closureX{\rect_2}$ is a point that is a common corner of both of them.
    \end{compactenumi}
\end{observation}

The following is an immediate consequence of the dual graph of the decomposition being planar, together with \thmref{number_pairs}.
\begin{lemma}
    \lemlab{number_edges}%
    We have $N = \sum_{\rect \in \cDecomp}\cardin{\Neighbors{\rect}} = O(|\cDecomp|)$.  If $\cP$ and $\cQ$ are $c$-packed, then $N = O(\frac{cn}{\eps})$, where $n=\cardin{P}+\cardin{Q}$.
\end{lemma}

\paragraph*{Computing the neighborhood for all rectangles.}

The task is to compute, for all $\rect \in \cDecomp$, and for all edges of $\rect$, a sorted list of its neighbors that share this edge.  We describe how to handle the vertical edges (a similar algorithm applies to the horizontal edges).

For a rectangle $\rect \in \cDecomp$, where $\rect = \IOCX{x_i,x_j} \times \IOCX{y_{i'},y_{j'}}$, consider the corresponding rectangle in the index space $\IRX{\cardin{P}} \times \IRX{\cardin{Q}}$, denoted by $\rect^{}_\ZZ = \IRY{i}{j} \times \IRY{i'}{j'}$.  The four corners of such a rectangle $\rect_\ZZ^{}$ are $(i, i'), (j, i'), (i, j')$ and $ (j,j')$.  Let $\LL$ be the list of all the vertices of all the rectangles in $\cDecomp$ (the list $\LL$ stores for each vertex the rectangle it came from).

Let $ n =\cardin{P} + \cardin{Q}$. The list $\LL$ can be sorted in $O( n + \cardin{\LL})$ time using radix sort. Indeed, the vertex $(i,j) \in \LL$ can be mapped to the integer $(n+1) i + j$, and these numbers all lie in $\IRX{\pth{n+1}^2}$, and thus radix sort is directly applicable. In the resulting order, all vertices sharing the same $x$-axis are grouped, sorted internally by increasing $y$-order (note that $\LL$ might contain the same vertex at most four times, associated with four different rectangles).

Consider the slice $\LL_i$ of the sorted list for all vertices sharing the same $x$-axis with value $i$. The algorithm scans the list $\LL_i$ with increasing $y$ value (i.e., the algorithm sweeps up on the vertical line $x=i$). At each step, the algorithm tracks the (at most four) rectangles on the left and right containing the current point $(i,y)$. The algorithm records the adjacency between these rectangles as appropriate. It continues this vertical scanning process by handling the next value stored in $\LL_i$ (recording any new adjacency discovered as the set of active rectangles changes). This takes $O(\cardin{\LL_i})$ time for $\LL_i$, and $O( \cardin{\LL})$ time overall.

The same algorithm is repeated to record horizontal adjacencies. At the end of this process, every rectangle has four \emph{sorted} lists of its neighboring rectangles, corresponding to the four edges of the rectangle.

\begin{corollary}%
    \corlab{neighbors}%
    If $\cP$ and $\cQ$ are $c$-packed and have $n$ vertices in total, then the sets~$\Neighbors{\rect}$ for all rectangles $\rect \in \cDecomp$, are computed in~$O\left(\frac{cn}{\eps}\right)$ time, by the algorithm described above.
\end{corollary}

\begin{remark}
    \remlab{direct}%
    Treating the above neighboring relationship as a graph, one can orient the edges computed by the above algorithm so that they are always oriented to the right (across a vertical edge), to the top (if they cross a horizontal edge), and diagonally to the top-right if they are across a corner (we ignore the bottom-right corner edges in this case). This creates a natural \DAG on $\cDecomp$ which complies with the monotonicity constraints for the strong variants of the \Frechet distance.
\end{remark}

\section{Approximating the weak variants of the \Frechet distance}\seclab{frechet}

\subsection{Warm up: weak discrete \Frechet distance}%
\seclab{weak-discrete}%

As a reminder, the weak discrete \Frechet distance between the point sequences $P$ and $Q$ corresponds to the smallest pass (i.e., the highest elevation in the alignment) that any alignment of $P$ and $Q$ encounters while traversing the grid $\PQ$, see \secref{discrete_f}.  The algorithm of \secref{decomposing-free-space} computes a decomposition of the grid~$\PQ$ into $\OO(\frac{cn}{\eps})$ rectangles $\Decomp$, such that inside each (discrete) rectangle the elevation function is the same up to a factor of $1 \pm \eps$, see \remref{wspd_like} \itemref{its_a_constant}. Specifically, for $\rectY{\vp}{\vq} \in \Decomp$ the simplified elevation function is $\dY{\repX{\vp}}{\repX{\vq}}$. Construct the dual graph to the decomposition~$\Decomp$ --- this is an undirected graph where each rectangle in $\Decomp$ induces a vertex, and there are edges between neighboring rectangles. The weight of an edge $e = \{\rectY{\vp}{\vq},\rectY{\vp'}{\vq'}\}$ in this graph is $\max\{\dY{\repX{\vp}}{\repX{\vq}}, \dY{\repX{\vp'}}{\repX{\vq'}}\}$.  To compute the (weak) \Frechet distance, one has to compute the bottleneck shortest-path between the node corresponding to the bottom-left rectangle and the top-right rectangle in $\Decomp$. This can be done in linear time in the size of the undirected graph, by a classical prune-and-search algorithm, see \cite{kp-bspp-06}. \corref{neighbors} implies the following result.

\begin{theorem}
    \thmlab{w_discrete_frechet}%
    Given two $c$-packed sequences of points~$P$ and~$Q$ with $n$ vertices in total, and a parameter~$\eps \in (0,1)$, the above algorithm computes, in $\OO( \tfrac{c n}{\eps})$ time, a $(1+\eps)$-approximation to the discrete weak \Frechet distance between~$P$ and~$Q$.
\end{theorem}

\subsection{Adding continuity in the weak setting}%
\seclab{weak-continuous}%

In the continuous case, the rectangles in $\cDecomp$ induce a decomposition of the domain $\cPQ$, see \secref{continuous}. On each rectangle, the simplified elevation function is a $(1+\eps)$-approximation to the true elevation function. Note that by design, the simplified elevation function is also convex within each rectangle. The task is to compute a path through $\cDecomp$ that has minimum pass, see \Eqref{pass}. We use the same dual graph construction as in \secref{weak-discrete}.  Two vertices $s, t$ are added to this graph, where $s$ (resp. $t$) is located at $(0,0)$ (resp. $(\lenX{P},\lenX{Q})$). The vertices $s$ and $t$ are each connected to the respective vertex that corresponds to the (unique) rectangle containing them. The weight of these two new edges is the elevation of the corner the respective vertices correspond to.

The edge weights of the other edges are defined as follows. Consider the boundary between two neighboring rectangles $\rect$ and $\rect'$, with their common boundary being a segment $s = \closureX{\rect} \cap \closureX{\rect'}$. This segment could be horizontal, vertical, or a single point.  Then, the weight of the corresponding edge $e$ in the graph between $\rect$ and $\rect'$ is
\[
    w(e)=\inf_{(x,y)\in s}\elevS(x,y).
\]
Observe that one can compute the weight in constant time per edge. As in the discrete case, computing the weak continuous \Frechet distance can be done by computing the bottleneck shortest path in this graph between $s$ and $t$.

\begin{theorem}
    \thmlab{w_cont_frechet}%
    Given two polygonal curves $\cP$ and $\cQ$ that are $c$-packed, and a parameter $\eps \in (0,1)$, one can compute a $(1+\eps)$-approximation to the continuous weak \Frechet distance between~$\cP$ and~$\cQ$, in $\OO( \tfrac{c n}{\eps})$ time, where $n$ is the total number of vertices of $P$ and $Q$.
\end{theorem}
\begin{proof}
    First, run the above algorithm with $\eps'=\eps/3$.  The running time is bounded by $\OO( \tfrac{c n}{\eps})$ by the same argument as in the discrete case, as the weight is computed in constant time per edge.  As for correctness, consider an optimal alignment $\alignment^*$ through the domain $\cPQ$, where the optimality is with respect to the elevation function $\elev$.  Let $S_{\alignment^*}$ be the $s$-$t$ path in the graph by tracking (in order) the rectangles $\alignment^*$ visits (this path might contain cycles, which are shortcut to keep it simple).  For an edge $e \in S_{\alignment^*}$, the alignment $\alignment^*$ encounters at least height $w(e)$ when traversing the corresponding (dual) boundary edge between the two rectangles (using the simplified elevation function). Also, let $\aprxPath$ denote a (bottleneck) shortest $s$-$t$-path in the graph (as a reminder, the weights of edges are computed using the approximate elevation function $\elevS$).

    Denote by $\passY{\Path}{w}=\max_{e\in \Path}w(e)$ the bottleneck weight of a path $\Path$ in the graph, and let $\passY{\alignment}{f} = \max_{ p \in \pi} f(p)$ denote the maximum value of an elevation function $f$ on an alignment~$\alignment$ (i.e., its the pass of $\alignment$ under $f$). With the observation from above, we have
    \[
        \passY{\aprxPath}{w}%
        \leq%
        \passY{S_{\alignment^*}}{w} \leq \passY{\alignment^*}{\elevS} \leq (1+\eps')\passY{\alignment^*}{\elev}\leq(1+\eps)\passY{\alignment^*}{\elev}.
    \]

    For each edge $e$ in the graph, the \emph{portal}, on the corresponding boundary segment $s_e$ between the two corresponding rectangles in $\cPQ$, is a point on which $w(e)$ is attained, i.e.\ $\elevS(x,y)=w(e)$ (i.e., it is a point realizing the minimum of $\elevS$ on $\closureX{s_e}$). The portal for edges adjacent to the start $s$ is the point $(0,0)$, and the portal for edges adjacent to the target $t$ is the upper-right point of $\cPQ$.  A minor technicality is that there might not be a point $(x,y)\in s$ with $\elevS(x,y)=w(e)$, as this point might be realized by the endpoint of the segment that is outside this (partially) open set. For the sake of clarity of exposition, we ignore this technicality, as the argument can easily be extended to handle this.

    With the portals, the path $\aprxPath$ can be reinterpreted as an alignment $\alignment_{\aprxPath}$ by connecting consecutive portals on the path with straight segments. We have
    \[(1-\eps)\passY{\alignment^*}{\elev}%
        \leq%
        (1-3\eps')\passY{\alignment_{\aprxPath}^{}}{\elev}%
        \leq%
        \frac{(1-\eps')^2}{(1+\eps')}\passY{\alignment^{}_{\aprxPath}}{\elev} \leq%
        \frac{(1-\eps')}{(1+\eps')} \passY{\alignment^{}_{\aprxPath}}{\elevS}%
        \leq%
        \passY{\aprxPath}{w}.\]

    Indeed, the simplified elevation function $\elevS$ is a $(1+\eps')$-approximation of $\elev$, which is a continuous function. Moreover, it is convex within each rectangle and (by virtue of $\gamma$ being continuous) is approximately convex within the closure of each rectangle.
\end{proof}

\section{Approximating the strong variants of the Fréchet distance}
\seclab{strong}

Next, we describe the algorithms for the strong variants of the \Frechet distance, starting with the discrete setting. To keep the presentation concise and consistent, the text relies on the continuous decomposition $\cDecomp=\Set{\rectY{\vp}{\vq}}{ (\vp,\vq) \text{ leaf of }\Tcont }$, which provides a disjoint cover of the domain $\cPQ$ (see \secref{decomposing-free-space}). While this description targets the continuous setting, the simplified elevation functions of these rectangles are constant in the discrete case (i.e., when restricted to the relevant grid points in the rectangle). To unify the presentation, we canonically extend the constant simplified elevation function of any discrete rectangle $\rect\in\Decomp$ to the support $\suppX{\rect}$ of the rectangle, i.e.\ the corresponding continuous rectangle in $\cDecomp$. %

\paragraph{Notation.}
We introduce some more notation to unify the description of the discrete and continuous variants.  For $x,x'\in \cXP = [0, \lenX{P}]$, and the interval $I = [x,x']$, let $\XP[x,x'] = (x_i,\ldots,x_j)$ denote the subsequence of $\XP$ that lies in $I$. The corresponding subsequence of $P$ is $P[x,x'] = (p_i,\ldots,p_j)$.  The continuous counterpart, the subcurve of $\cP$ from $P(x)$ to $P(x')$ is denoted by $\cP[x,x']{\Bigr.}$. Observe that the endpoints of $P[x,x']$ might be in the interior of $[x,x']$, if $x$ or $x'$ are not in $\XP$.  In the following, $\EitherX{\cdot}$ is used to denote an entity that can be either discrete or continuous. For example, $\EP$ is either $P$ or~$\cP$, depending on the context.

\paragraph{Cost function.}
Given two point sequences~$P$ and~$Q$. The task at hand is to approximate, for all $x \in \EXP, y\in \EYQ$, the cost function
\begin{equation*}
    \psi(x,y)
    =%
    \dFY{\EP[0,x]}{\EQ[0,y]}.
\end{equation*}
It turns out that it is enough to approximate this function on the boundaries of the rectangles of the decomposition $\cDecomp$, as was done in the weak case.

\subsection{The main algorithm}
\seclab{mainAlgorithm}

\subsubsection{Algorithm overview}
The idea is to process the rectangles one by one in a topologically sorted order and, in each step, compute the cost function on the boundary of the current rectangle. Initially, the algorithm approximates the cost function on the bottom and left sides of $\cPQ$.  Let $\rect_1, \ldots, \rect_t$ be a topological ordering of the rectangles of $\Decomp$.  We maintain the cost function along the \emph{propagation front} that separates the processed rectangles from the unprocessed rectangles, see \figref{sweep-boundary}.

\begin{figure}
    \centering%
    \includegraphics[page=1]{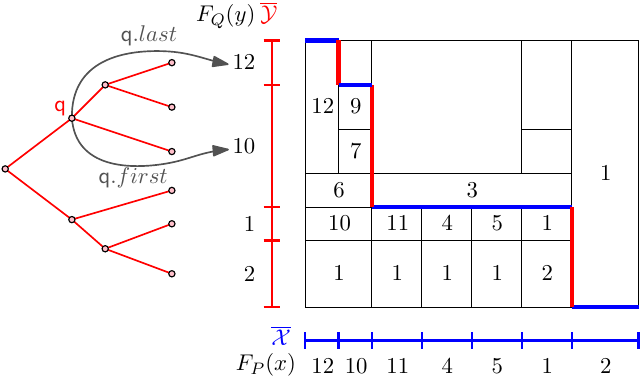}
    \caption{The figure shows the propagation front within the domain of the elevation function. On the axes of~$\cPQ=\cXP\times\cYQ$, the corresponding functions~$F_P$ and~$F_Q$ are shown on the intervals~$\cXP$ and~$\cYQ$. On the left, the figure shows (a subtree of) the tree $\Tq$ and an example of the pointers $\last{\vq}$ and $\first{\vq}$ pointing to the respective part of the function $F_Q$.}
    \figlab{sweep-boundary}
\end{figure}

In the beginning of $i$\th iteration, the algorithm has access to some representation of the computed cost function on the bottom and left edges of $\rect_i$, and the (constant complexity) approximate elevation function in $\rect_i$. Using the two, carefully, as they are not of constant size, the algorithm computes the cost function on the top and right edges of $\rect_i$. The algorithm continues in this process till all the rectangles are handled.  Throughout this process, the algorithm stores the cost function along the propagation front as a piecewise function using a doubly linked list. For technical reasons, this is done by storing the function on the horizontal pieces and the vertical pieces separately, resulting in two lists.  In particular, the trees~$\Tp$ and~$\Tq$ from \secref{decomposing-free-space} are used to store pointers that will be used to navigate from a rectangle boundary to the relevant part of the two lists. We will update these pointers in a lazy fashion.

\subsubsection{A central observation}

The following observation states a basic property that we use to update the propagation front when processing a rectangle.

\begin{observation}%
    \obslab{concat}%
    For two point sequences~$P$ and~$Q$, let $x_0, x\in \EXP$, and $y_0,y\in \EYQ$, with $x_0\leq x$ and $y_0\leq y$. Then, for the strong (discrete or continuous) \Frechet distance, it holds that
    \begin{align*}
        \dFY{\EP[0,x]}{\EQ[0,y]}
        =%
        \min_{(x',y')\in  U\cup W}
        \max\Bigl(%
        \dFY{\bigl.\smash{\EP[0,x']}}{\smash{\EQ[0,y']}}, \,
        \dFY{\bigl.\smash{\EP[x',x]}}{\smash{\EQ[y',y]}}
        \Bigr),
    \end{align*}
    where $U=[x_0,x]\times y_0$ and $W=x_0\times[y_0,y]$, see \figref{concat}.
\end{observation}

In particular, when considering how to process rectangle $\rect_i$, we observe the following: (1) the constant elevation of $\rect_i$ is a lower bound for the cost function on the upper and right boundaries of $\rect_i$; (2) beyond this lower bound, the minimum of the cost function along the left boundary is an upper bound on the cost function along the top boundary (a similar observation holds for the bottom and right boundary); (3) between these two bounds, the cost function along the top boundary is decreasing from left to right (and similarly along the right boundary, from bottom to top), since the path can cross inside $\rect_i$ to avoid the cost from the bottom (resp., left) boundary; refer to \figref{ex_steps_1_3} for an example of the cost function before and after processing a rectangle. The algorithm for updating the cost function is described in detail further below.

\begin{figure}
    \centering%
    \includegraphics{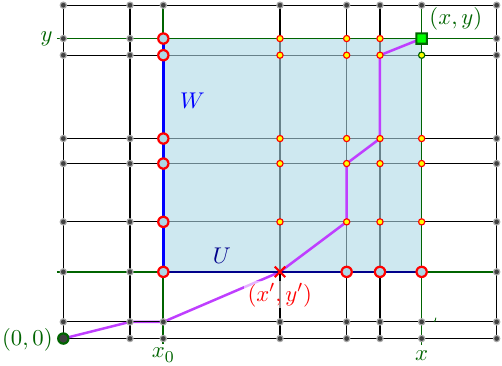}
    \caption{Illustration of \obsref{concat} in the discrete case.}
    \figlab{concat}
\end{figure}

\subsubsection{The algorithm in detail}

\paragraph{Data structure.} We use two lists, where the list $L_P$ stores the cost function of the horizontal pieces of the propagation front and $L_Q$ stores the cost function of the vertical pieces of the propagation front.  An element $a$ of~$L_P$ (resp.\ $L_Q$) stores a simplified elevation function~$f_a: I_a\rightarrow \Re_{\geq 0}$ defined on an interval~$\IVX{a}\subseteq\cXP$ (resp.\ $\IVX{a}\subseteq\cYQ$) The intervals~$I_a$ with~$a\in L_P$ (resp.~$a\in L_Q$) are disjoint, and their union is~$\cXP$ (resp.~$\cYQ$).

For every node~$\vp$ in the tree~$\Tp$ (resp.~$\Tq$), we store two pointers $\first{\vp}$ and $\last{\vp}$ that either point to an element in the list~$L_P$ (resp.~$L_Q$) or to an element that is no longer part of the list~$L_P$ (resp.~$L_Q$).  If the pointer~$\first{\vp}$ (resp. $\last{\vp}$) points to an element~$a$ in one of the lists, then the start (resp.\ end) of the interval~$\InodeX{\vp}$ corresponding to~$\vp$ lies in $\IVX{a}$.  These pointers are used to locate the start/end of the interval corresponding to~$\vp$ in the list.

We denote the cost function stored for the horizontal front as~$F_P:\cXP\rightarrow\Re_{\geq 0}$ where $F_P(x)=f_a(x)$ with $x\in \IVX{a}$.  Analogously, we denote $F_Q:\cYQ\rightarrow\Re_{\geq 0}$ for the vertical front. In addition, for every corner point~$(x,y)$ of any closure of a rectangle that shares a boundary with the propagation front, we store in $\distY{x}{y}$ a $(1+\eps)$-approximation to~$\dFdY{\EitherX{P}[0, x]}{\EitherX{Q}[0,y]}$. Observe that later in the construction it holds that~$\cP(x)$ and~$\cQ(y)$ are vertices in both the discrete and the continuous setting.

\paragraph*{Initialization.}

We initialize the cost function on the bottom and left edges of the domain $\cPQ$ as follows. We initialize~$L_P$ such that $F_P(x)=\max\{\esX{x', 0}\mid x'\in [0, x]\}$, where $\esX{\cdot, \cdot}$ is the simplified elevation function computed in \secref{decomposing-free-space}. In the discrete case, $F_P(x)$ is a piecewise constant function, and for each constant piece we add a list element to~$L_P$. In the continuous case, $F_P(x)$ might be decomposed into constant pieces and pieces of the form~$\|\cP(x)-u\|$. Then, we add one list element for each of these pieces.

\paragraph*{Processing a rectangle.}  We process a rectangle $\rectY{\vp}{\vq}=\IOCX{x_0, x_1} \times \IOCX{y_0, y_1}$ in three steps. In each step, we execute the algorithm first for~$\EP$ and then for~$\EQ$.
Step~1 and Step~2 are always executed. For~$\EP$ (resp.~$\EQ$), Step 3 is only executed if the number of upper (resp.\ right) neighboring rectangles of~$\rectY{\vp}{\vq}$ is greater than~$1$.  In the following, we describe the steps for~$\EP$. The steps for~$\EQ$ are symmetric. See \figref{ex_steps_1_3} for an example.

\subparagraph{Step 1.}  We start by deleting elements from the list $L_P$ such that the function~$F_P$ restricted to the interval~$\IOCX{x_0, x_1}$ is descending.  Let $\vp_1, \vp_2, \dots, \vp_k$ be the nodes in $\Tp$ corresponding to the lower neighboring rectangles of $\rectY{\vp}{\vq}$ in order.  We delete elements from the list such that afterwards~$F_P$ restricted to~$\IOCX{x_0, x_1}$ realizes the following function $f_1(x)=\min\{F_P(x')\mid x_0\leq x'\leq x\}$.  This can be done efficiently by using the pointers of~$\vp_i$, since in the discrete case $F_P$ restricted to $\InodeX{\vp_i}$ is descending, and in the continuous case, it is a valley function, which we will handle in more detail in \secref{strong-continuous}.  Further, we update the pointers of $\vp_1, \dots, \vp_k$ to the new first and last entries of the list that correspond to their interval.  Then, if $k> 1$, we consolidate the pointers from the nodes $\vp_1, \dots, \vp_k$ of tree~$\Tp$ upwards to the node~$\vp$ along the paths towards the root. More specifically, we use the following update rule. Let $w$ be a node visited in this process, and let its left child be $w_1$, and its right child be $w_2$. We set $\first{w}$ to $\first{w_1}$ and $\last{w}$ to~$\last{w_2}$.

\subparagraph*{Step 2.}
Define~$d$ as the minimum of the value of~$\last{\vq}$ (before Step~2) and $\distY{x_0}{y_0}$.  We modify the list such that~$F_P$ restricted to~$\IOCX{x_0, x_1}$ is~$f_2(x)=\min\{d, f_1(x)\}$. Note that we only have to cut off elements from the beginning of the list, starting at $\first{\vp}$ and possibly adding one new element.  Then, update the pointers of $\vp$ accordingly.  Next, we modify the list again such that~$F_P$ restricted to~$\IOCX{x_0, x_1}$ is~$f_3(x)=\max\{f_2(x), \esX{x, y_1}\}$. Note that if $\esX{x, y_1}$ restricted to~$\rectY{\vp}{\vq}$ is constant, we only have to cut off elements from the end of the list ending at $\last{\vp}$ and possibly add one new element. The case that $\esX{x,y_1}$ is non-constant is handled in \secref{strong-continuous}. For every~$x$, such that $(x, y_1)\in \rectY{\vp}{\vq}$ is a corner point of a rectangle in~${\cDecomp}$, we set~$\distY{x}{y_1}=f_3(x)$.

\subparagraph*{Step 3.}  The third step pushes the pointers down along the paths from the node $\vp$ to the nodes $\vp'_1, \dots, \vp'_{k'}$ of~$\Tp$ corresponding to the rectangles on the top boundary of~$\rectY{\vp}{\vq}$. Here, we update pointers to elements deleted from the list~$L_P$ as follows. Let~$v$ be a node visited in this process and let $v'$ be its parent. If the element~$\first{v}$ (resp.~$\last{v}$) is deleted of~$L_P$, we set $\first{v}$ (resp.~$\last{v}$) to the element of~$L_P$ whose interval contains the start (resp.\ end) of the interval corresponding to~$v$. By construction, this is either~$\first{v'}$ or~$\last{v'}$.

\begin{figure}
    \phantom{}\hfill%
    \includegraphics[page=2]{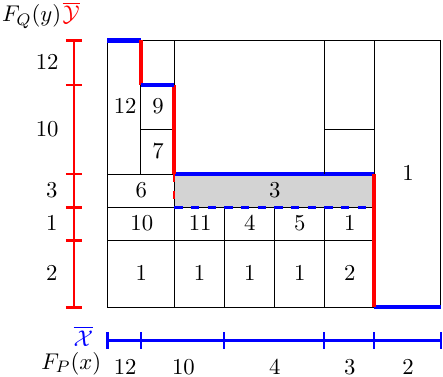}%
    \hfill%
    \includegraphics[page=3]{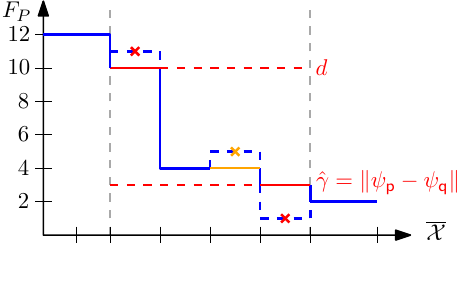} \hfill\phantom{}%
    \caption{This is the same example as in Figure~\figref{sweep-boundary} after the rectangle with elevation $3$ has been processed. On the right, the blue function shows $F_P$. The blue dashed pieces are removed when processing the rectangle with elevation $3$. The orange change is due to Step $1$, the red changes are due to Step $2$.}
    \figlab{ex_steps_1_3}
\end{figure}

\subsection{Analysis for the discrete Fréchet distance}

\begin{observation}
    \obslab{pointers_to_list} During the algorithm, the following remains true.
    \begin{compactenumi}
        \item The intervals~$\IVX{a}$ of the list elements~$a$ of $L_P$ (resp. $L_Q$) are disjoint.
        \item For every node~$v\in \Tp$ (resp. $v\in \Tq$), the start of the interval~$\InodeX{v}$ lies in $\IVX{\first{v}}$ and the end of the interval~$\InodeX{v}$ lies in $\IVX{\last{v}}$.
        \item Let $\rectY{\vp}{\vq}$ and $\rectY{\vp'}{\vq'}$ be neighboring rectangles with~$\InodeX{\vp}\subseteq \InodeX{\vp'}$ (resp.~$\InodeX{\vq}\subseteq \InodeX{\vq'}$). If their common boundary
        is part of the propagation front, then the pointers~$\first{\vp}$ and~$\last{\vp}$ (resp. $\first{\vq}$ and $\last{\vq}$) are contained in the list~$L_P$ (resp.~$L_Q$).
    \end{compactenumi}
\end{observation}

\begin{lemma}
    \lemlab{strong_correctness} During the algorithm in the discrete case, the following is true. Let~$(x,y)$ be a point on the propagation front. It holds that~$\dFdY{P[0, x]}{Q[0, y]}\in [(1-4\eps)d, (1+4\eps)d]$, where $d$ is defined as follows:
    \begin{compactenumi}
        \item If~$(x,y)$ is a corner point of a rectangle in~$\Decomp$, then define $d=\distY{x}{y}$.
        \item If~$(x,y)$ lies on the interior of a horizontal piece, then define~$d=F_P(x)$.
        \item Otherwise $(x,y)$ lies on the interior of a vertical piece, then define~$d=F_Q(y)$.
    \end{compactenumi}
\end{lemma}
\begin{proof}
    We prove this lemma by induction on the number of processed rectangles. Observe that after initialization, the lemma holds by \lemref{cont_fs_approx}.  Denote with $F_P$ the stored function before the rectangle~$\rectY{\vp}{\vq}$ gets processed and~$F_P'$ after~$\rectY{\vp}{\vq}$ is processed. Denote with $(x_0, y_0)$ the bottom left corner of the closure of~$\rectY{\vp}{\vq}$. Further define $U=[x_0,x]\times y_0$ and $W=x_0\times[y_0,y]$. By~\obsref{concat}, it holds that
    \begin{align*}
        \dFdY{P[0,x]}{Q[0,y]}=\min_{(x',y')\in  U\cup W}\max(\dFdY{P[0,x']}{Q[0,y']},\dFdY{P[x',x]}{Q[y',y]}).
    \end{align*}
    Let~$(x',y')$ be a point on the boundary of $\rectY{\vp}{\vq}$ with $x'\leq x$ and $y'\leq y$. By convexity inside a cell and \lemref{cont_fs_approx}, it holds that
    \begin{align*}
        \max\{\elevY{x''}{y''}\mid (x'', y'')\in [x',x]\times [y', y]\} %
        \leq \frac{1}{(1-\eps)}\max\left(\frac{(1+\eps)}{(1-\eps)}\esX{x',y'}, \esX{x, y}\right)
    \end{align*}
    Hence, $\max\{\esX{x',y'}, \esX{x, y}\}$ is a $(1+4\eps)$-approximation to~$\dFdY{P[x', x]}{Q[y', y]}$.

    Further, it holds that $\dFdY{P[0, x']}{Q[0, y']}\geq (1-\eps)\esX{x',y'}$. By the algorithm, it holds that.
    \[
        F_P'(x)=\max\{\esX{x,y}, \min\left(\{F_P(x')\mid x'\in [x_0, x]\}\cup \{F_Q(y')\mid y'\in [y_0, y]\}\right)\},
    \]
    since $\min\{F_Q(y')\mid y'\in [y_0, y]\}$ is the minium of the value of~$\last{\vq}$ (before Step~2) and $\distY{x_0}{y_0}$.  Hence, if the lemma statement holds before processing $\rectY{\vp}{\vq}$, then it also holds after processing it.
\end{proof}

\begin{theorem}
    \thmlab{s_discrete_frechet}%
    Let $P$ and $Q$ be $c$-packed sequences with a total of~$n$ vertices. For any $\varepsilon\in(0,1)$, there exists a $(1+\varepsilon)$-approximation algorithm for the strong discrete \Frechet distance of $P$ and $Q$ that runs in $\OO\left(\frac{cn}{\varepsilon}\right)$ time.
\end{theorem}
\begin{proof}
    The correctness follows from \lemref{strong_correctness} and running the algorithm with $\eps'=\eps/4$.

    It remains to prove the running time. We process the rectangles in a topological order such that at the time a rectangle~$\rect$ is processed, all neighboring rectangles below and to the left have already been processed. Such a topological ordering exists and can be computed in~$\OO(|\cDecomp|)=\OO(\frac{cn}{\eps})$ time by~\cite{k-tsln-62} and~\thmref{number_pairs}.  The initialization runs in~$\OO(\frac{cn}{\eps})$ time as well. Now consider Step~1 for~$P$, when processing a rectangle~$\rectY{\vp}{\vq}$. By \obsref{pointers_to_list}, $F_P$ restricted to $\InodeX{\vp_i}$ is monotone decreasing for all~$i$ and $\first{\vp_i}$ and $\last{\vp_i}$ point to the start and the end of~$\InodeX{\vp_i}$ in the list~$L_P$. Hence, the updates to the list can be done in~$\OO(r_1)$ time, where~$r_1$ is the number of deleted list elements in Step~1. The rest of Step~1 needs~$\OO(\Neighbors{\vp})$ time, since this bounds the size of the visited nodes of~$T_P$. Step~$2$ runs in~$\OO(r_2)$ time, where $r_2$ is the number of deleted list elements in Step~$2$. Similar to the last part of Step~1, Step~3 runs in~$\OO(\Neighbors{\vp})$ time. Since every element gets deleted at most once and we only add a constant number of pieces for each rectangle, together with \thmref{number_pairs} and~\lemref{number_edges}, it follows that the total running time is in $\OO\left(\frac{cn}{\varepsilon}\right)$.
\end{proof}

We can improve the running time of \thmref{s_discrete_frechet} in the following way. We first compute a constant factor approximation in~$\OO(cn)$ time using \thmref{s_discrete_frechet} (e.g., for $\eps=1/2$) Then, to get a $(1+\eps)$-approximation, we use \thmref{binary_search} and the stable $(1+\eps)$-approximate decision algorithm by Bringmann and \Kunnemann~\cite{bk-iafdc-17}, that runs in $\OO(\frac{cn}{\sqrt{\eps}}\log \tfrac{1}{\eps} )$ time for the strong discrete \Frechet distance.

\begin{corollary}
    \corlab{s_discrete_frechet_faster}%
    Let~$P$ and~$Q$ be $c$-packed curves with $n$ vertices total. For any $\varepsilon\in(0,1)$, there exists a $(1+\varepsilon)$-approximation algorithm for the strong discrete \Frechet distance of $P$ and $Q$ in $\OO(\frac{cn}{\sqrt{\varepsilon}} \log \tfrac{1}{\eps} )$ time.
\end{corollary}

\subsection{The full setting: monotone and continuous}\seclab{strong-continuous}
The algorithm of \secref{mainAlgorithm} is correct for the (strong continuous) \Frechet distance as well. However, its running time is too high for our purpose. Therefore, we modify the algorithm to obtain a $(\sqrt{2}+\eps)$-approximation for the (strong continuous) \Frechet distance in $\OO(cn/\eps)$ time. In particular, we slightly modify the simplified elevation function to ensure that the running time of maintaining the cost function does not increase.

\paragraph*{Challenges.} The challenge in the continuous case is the following. It can happen that the simplified elevation function for a rectangle $\rectY{\vp}{\vq}$ is not constant.  In Step~2, the algorithm computes the intersection points of the simplified elevation function of a rectangle and the so far computed function~$F_P$ (resp.~$F_Q$).  To be able to charge this computation time to the rectangles processed so far, we bounded the number of intersection points of these two functions. However, if the simplified elevation function is of the form $\normX{\cP(x)-u}$ or $\normX{\cQ(y)-u}$, there could be many intersection points (see \figref{Continuous_without_root2} for an example).

\begin{figure}
    \centering%
    \includegraphics[page=2]{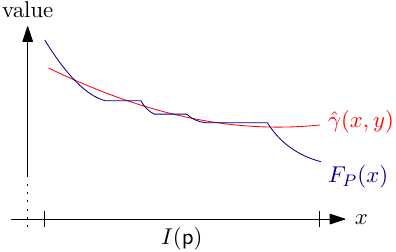}%
    \caption{The so far computed function~$F_P(x)$ and the simplified elevation function~$\esX{x,y}$ can have multiple intersection points.}
    \figlab{Continuous_without_root2}
\end{figure}

\paragraph{Idea.} In order to avoid this, we start by constructing $\sqrt{2}$-approximations for the non-constant simplified elevation functions.  Also, we observe that in the case that the simplified elevation function for a rectangle $\rectY{\vp}{\vq}$ is not constant, $\vp$ or~$\vq$ must be a single edge and hence a leaf of~$\Tp$ or~$\Tq$.
Later, we show that using the approximated functions, we can compute the upper envelope of~$F_P(x)$ and the approximated simplified elevation functions efficiently.

Note that the algorithm uses the simplified elevation function only on the boundary of rectangles. Hence, it is enough to approximate the simplified elevation functions on the boundaries of the rectangles.
\begin{lemma}
    \lemlab{aprx_linear}%
    If the simplified elevation function~$\esX{x,y}=\normX{\cP(x)-u}$ (resp.~$\normX{\cQ(y)-u}$) on a boundary of $\rectY{\vp}{\vq}$ for some~$u\in \Re^d$, then it can be $\sqrt{2}$-approximated on this boundary by a function~$\elev'(x,y) = |x-a|+b$ (resp.~$\elev'(x,y) = |y-a|+b$) for some~$a,b\in \Re$.
\end{lemma}
\begin{proof}
    By translation and rotation, we can assume $u$ is at the origin, and $\cP(x) = (b, a+x)$ (all other dimensions are zero and can be ignored), for some real numbers $a,b$. Thus, $\esX{x,y} = \normX{\cP(x)}_2 \leq \normX{\cP(x)}_1 = |a+x| + |b| \leq \sqrt{2} \normX{\cP(x)}_2 = \sqrt{2}\, \esX{x,y}$, as $\normX{\cdot}_1 \leq \sqrt{d}\normX{\cdot}_2$ in $\Re^d$.%
\end{proof}

Now, we run the algorithm of \secref{mainAlgorithm} with the {$\sqrt{2}$-approximations} for the non-constant simplified elevation functions of~\lemref{aprx_linear}.

\begin{figure}[ht]
    \centering%
    \includegraphics{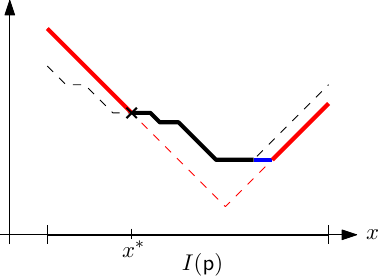}
    \caption{The black function was the value function before processing $\rectY{\vp}{\vq}$. The red function is the $\sqrt{2}$-approximate simplified elevation function~$\elev'(\cdot, \cdot)$ for~$\rectY{\vp}{\vq}$. The new function is drawn thicker.}
    \figlab{funct_strongContinuous}
\end{figure}

The following observation is crucial in the runtime analysis and can be proven easily via induction.
\begin{observation}
    \obslab{edge_function_property}%
    Let~$\vp$ be a leaf of~$\Tp$. Then, it holds that~$F_P$ restricted to the interval $\InodeX{\vp}$ is a piecewise linear valley function such that the slope of the pieces is in $\{-1, 0\}$ except for the last piece whose slope is~$0$ or~$+1$.
\end{observation}

\begin{lemma}
    \lemlab{sqrt2_approx}%
    Let $\cP$ and $\cQ$ be $c$-packed. For any $\varepsilon>0$, there exists a $(\sqrt{2}+\eps)$-approximation algorithm for the strong continuous \Frechet distance of $\cP$ and $\cQ$ in $\OO\left(\frac{cn}{\varepsilon}\right)$ time.
\end{lemma}
\begin{proof}
    We use the algorithm of \secref{mainAlgorithm} but with the $\sqrt{2}$-approximation of the simplified elevation function on the non-constant pieces of $F_P$. Note that \lemref{strong_correctness} and its proof still stand in the continuous case using the $\sqrt{2}$-approximations of the non-constant simplified elevation functions from \lemref{aprx_linear}. The only thing that changes is that we have $\dFY{\cP[0,x]}{\cQ[0,y]}\in[(\sqrt{2}-4\eps)d,(\sqrt{2}+4\eps)d]$ in the statement of the lemma.  So, the correctness follows from this adapted version of \lemref{strong_correctness} together with \lemref{aprx_linear}.

    It remains to prove the bound on the running time.  Consider Step~1 for~$\cP$, when processing a rectangle~$\rectY{\vp}{\vq}$ Let $\vp_1,\ldots,\vp_k$ be the nodes in $\Tp$ corresponding to the lower neighboring rectangles of $\rectY{\vp}{\vq}$. Then, by \obsref{edge_function_property}, it holds that~$F_P$ restricted on $\IVX{\vp_i}$ for $i=1, \dots, k$ is monotone decreasing or a valley function. Hence, using the pointers, we can perform Step~$1$ in $\OO(k+r_1)$ time, where $r_1$ is the number of deleted elements from the list in Step~1.  If $\vp$ is not an edge, Step~$2$ for~$\cP$ takes linear time in the number of deleted elements from the list, since $f_1$ is monotone decreasing. If~$\vp$ is an edge, then~$f_2$ has the properties of \obsref{edge_function_property} (see \figref{funct_strongContinuous}).  Let $x^*$ be the smallest value such that $f_2(x^*)=\elev'(x^*, y_1)$.  By construction of~$f_2$, for all~$x\geq x^*$ it holds that $f_2(x)\geq f_2(x^*) +(x^*-x)$. Hence, there cannot be more than two intersection points of~$f_2$ with~$\elev'(x,y_1)$ since the slope of~$\elev'$ is first~$-1$ and then~$+1$. Further, these two points can be found in linear time in the number of deleted list elements. This bounds the running time of Step~2.  Step~3 for~$\cP$ uses~$\OO(k')$ time, where~$k'$ is the number of rectangles to the top of~$\rectY{\vp}{\vq}$.  Since every element gets deleted at most once, and we only add a constant number of pieces for each rectangle, together with \lemref{number_edges}, the total running time is in $\OO\left(\frac{cn}{\varepsilon}\right)$.
\end{proof}

To get a $(1+\eps)$-approximation algorithm for the strong continuous \Frechet distance, we first compute a constant factor approximation in~$\OO(cn)$ time using \lemref{sqrt2_approx}. Then, the next theorem follows by \thmref{binary_search} and the stable $(1+\eps)$-approximate decision algorithm by Bringmann and \Kunnemann~\cite{bk-iafdc-17}, that runs in $\OO((cn/\sqrt{\eps})\log^2(1/\eps))$ time.

\begin{theorem}
    \thmlab{s_cont_frechet}%
    Let $\cP$ and $\cQ$ be $c$-packed with a total of~$n$ vertices.  For any $\varepsilon>0$, there exists a $(1+\varepsilon)$-approximation algorithm for the strong continuous \Frechet distance of $\cP$ and $\cQ$ in $\OO((cn/\sqrt{\eps})\log^2(1/\eps))$ time.
\end{theorem}

Observe that if the set $\cDecomp$ consists of all cells, then we do not have to simplify the elevation functions as done in \secref{decomposing-free-space}. So, if we run the above algorithm on $\cDecomp$, we get a $\sqrt{2}$-approximation in $\OO(\cardin{\cDecomp})$ time.  In 1-dimension, we do not even have to approximate the elevation functions of edge-vertex and vertex-edge rectangles, and hence, the algorithm is exact in this scenario. If the complexity of $\cP$ is~$n$ and the complexity of $\cQ$ is~$m$, the size of $\cDecomp$ is~$\OO(nm)$. Hence, we get the following corollary.

\begin{corollary}
    \corlab{1dexact}%
    Let $\cP$ be a curve of complexity $n$ and $\cQ$ be a curve of complexity $m$.  There exists a $\sqrt{2}$-approximation algorithm for the strong continuous \Frechet distance for general curves in~$\Re^d$ that runs in $\OO(nm)$ time. In 1D, this algorithm is exact.
\end{corollary}

\section{Faster binary search via an approximate decider}
\seclab{faster}

In this section, we discuss different ways to turn a low-quality approximation into a $(1+\eps)$-approximation. Here, one is provided with a decision procedure (i.e., a decider), and the task is to use a few calls to this decider to boost the approximation to the desired quality.  This is done in such a way that the binary search incurs only a constant factor overhead in total, under some mild conditions on the running time of the decider.

The basic ideas used to derive \thmref{binary_search} below are not new, and the main purpose here is to present the basic settings and state the result so that it can be used in a plugin fashion. Surprisingly, it seems this result is not explicitly stated in previous work.

We have an input $P$ of size $n$, and the task is to approximate $\vOpt = f(P)$, for a function $f$, where this value is guaranteed to lie in a given interval $[\alpha, \beta]$, with $\beta > \alpha > 0$. Assume we are given an approximate decider.

\begin{defn}
    \deflab{decider}%
    Let $\tau > 0$ be a parameter. Given a set $P$ of size $n$, and a (non-negative) function $f$, a procedure $\decider$ is a \emphi{$(1+\tau)$-decider} for $\vOpt = f(P)$, if $\decider(P, r, 1+\tau)$ returns one of the following:
    \begin{compactenumi}
        \item $\vOpt \in [\alpha, (1+\tau)\alpha]$, where $\alpha$ is some real number,
        \item $\vOpt < r$, or
        \item $\vOpt > r$.
    \end{compactenumi}
\end{defn}
In the following, the decider running time is denoted by $T(n,1+\tau)$.

\begin{defn}
    For an interval $I = [\alpha, \beta]$, with $\beta \geq \alpha > 0$, its \emphi{spread} is the quantity $\spreadX{I} = \beta/\alpha$.
\end{defn}

Note, that if the algorithm computed an interval $I$ that contains $\vOpt$, and $\spreadX{I} \leq 1+\eps$, then this yields a $(1+\eps)$-approximation to $\vOpt$.

\paragraph*{A first attempt.}
To get a $(1+\eps)$-approximation using the decider, for $\eps \in (0,1)$, given an interval $[\alpha,\beta]$, the natural approach is to divide the interval with a middle value $(\alpha+\beta)/2$, call the decider and then recurse either left or right, repeating the process till one is left with an interval with spread at most $1+\eps$ which must contain the value $\vOpt$. This requires $ O( \log \frac{\beta}{\eps \alpha} ) = O( \log \tfrac{\beta}{\alpha} + \log \tfrac{1}{\eps} )$ calls to the decider, and the overall running time would be $O\bigl( T(n,1+\eps) ( \log \tfrac{\beta}{\alpha} + \log \tfrac{1}{\eps} ) \bigr)$.

\paragraph*{A second attempt.}
Consider all the approximate answers the binary search might return: $\alpha, (1+\eps)\alpha, \ldots, \pth{1+\eps}^i\alpha, \ldots, \beta$.  If the given interval is $[\alpha,\beta]$, then $\vOpt$ can be rounded to be one of the values
\begin{equation*}
    \alpha_i
    =
    \alpha \pth{1+\eps}^i,
    \qquad \text{for}\qquad i=1,\ldots,
    M = O(\log_{1+\eps}\tfrac{\beta}{\alpha} )
    =
    O( \tfrac{1}{\eps} \log \tfrac{\beta}{\alpha} ).
\end{equation*}
Thus, we can perform a direct binary search on these $M$ values.
This would require $\log M$ calls to the decider, resulting in running time $O\bigl( T(n,1+\eps) \log M \bigr) = O(T(n,1+\eps) ( \log \log \tfrac{\beta}{\alpha} + \log \tfrac{1}{\eps} ) \bigr)$.

\medskip%

One can do even better if the decider running time deteriorates gracefully, as the parameter $\eps$ decreases.

\begin{defn}
    \deflab{stable}%
    A decider is \emphi{stable}, if for any $\tau > \eps > 0$, its running time $T(n,1+\eps)$ for computing $(1+\eps)$-approximation, complies with the inequality
    \begin{math}
        \tfrac{T(n,1+\eps) }{T(n,1+\tau)} =%
        O\bigl( {( \tfrac{\tau}{\eps(\tau+1)} )}^c \bigr),
    \end{math}%
    where $c>0$ is some constant.
\end{defn}

Informally, if the decider is stable, running it with a large constant approximation factor is the same as running it with $2$ (in particular, $\tau$ can be much larger than $1$). For $\eps \in (0,1)$, the decider running time for $(1+\eps)$-approximations deteriorates at least geometrically with $\eps$.

\begin{theorem}
    \thmlab{binary_search}%
    Let $P$ be a set of size $n$ and let $\vOpt=f(P)$ be an unknown quantity for a (non-negative) function $f$. Suppose we are given a parameter $\eps \in (0,1)$, an interval $ [\alpha,\beta] \ni \vOpt$, and a \defrefY{stable}{stable} $(1+\tau)$-decider for $\vOpt$ that runs in $T(n,1+\tau)$ time and works for any $\tau > 0$. Then, one can compute an interval $I$, containing $\vOpt$, such that $\spreadX{I} \leq 1+\eps$ (i.e., a $(1+\eps)$-approximation to $\vOpt$). The running time of this algorithm is $O\bigl( T(n,1+\eps) + T(n,2) \log \log \tfrac{\beta}{\alpha} \bigr)$
\end{theorem}
\begin{proof}
    The algorithm is iterative, maintaining an interval $I_i = [\alpha_i,\beta_i]$ containing $\vOpt$, with $\alpha_1 = \alpha$ and $\beta_1=\beta$. Let $c_i = \spreadX{I_i} = \beta_i / \alpha_i$. If $c_i \leq 1+\eps$, then we are done. Otherwise, in the $i$\th iteration, let $r_i = \alpha_i \sqrt{c_i}$. The algorithm now calls the provided procedure \decider{}$(r_i, \sqrt{c_i})$:
    \begin{compactenumI}
        \item If the decider returns an interval $K_i$ containing $\vOpt$, then $\spreadX{K_i} \leq \sqrt{c_i}$. The algorithm sets $I_{i+1}=K_i$ and continues to the next iteration.

        \item If \decider returned that ``$\vOpt < r_i$'', then the search continues in the interval $I_{i+1} = [\alpha_i, r_{i}]$.

        \item Otherwise, it continues in the interval $I_{i+1} = [r_{i},\beta_i]$.
    \end{compactenumI}

    \smallskip

    Note that $I_{i+1}$ always contains $\vOpt$ by construction.  Clearly, for all $i$, we have that $c_{i+1} \leq \sqrt{c_i}$. The algorithm stops when the spread of the current interval is in the range $J_1 = \IOCX{1,1+\eps}$.  The spread in the previous iteration was in the interval $J_2 = \IOCX{1+\eps, \pth{1+\eps}^2 }$.  Thus, the spread in the $t$\th iteration before stopping is in the interval $J_t = \IOCX{b(t), b(t+1)}$, where $b(t) = \pth{1+\eps}^{2^{t-1}}$.

    This squaring behavior implies that the number of iterations when the spread is larger than $2$ is at most $\OO( \log \log \tfrac{\beta}{\alpha} )$. Each such iteration takes $\OO( T(n,2))$ by the assumption on the stability of the decider.

    To analyze the iterations when the spread of the current interval is smaller than $2$, observe that $1+\eps \geq \exp(\eps/2)$ since $\eps \in (0,1)$, and as such $b(t) \geq \exp( \tfrac{\eps}{2} \cdot 2^{t-1})$. Thus, for $t=1 + \log_2 \tfrac{2}{\eps}$, $b( t ) \geq e>2$ and hence, the number of iterations with a small spread is at most $t$. Thus, overall, the algorithm performs $\OO( \log \frac{1}{\eps} + \log \log \tfrac{\beta}{\alpha} )$ iterations.

    Assume the iterations of the algorithm where the spread was smaller than $2$ were iterations $m, \ldots, M$, where $M$ is the number of the last iteration.  Let $\tau_i = c_i-1$, for $i=m,\ldots, M$. Observe that $\tau_M = \Theta(\eps)$, and the values of $\tau_m, \ldots, \tau_M$ behave like a geometric series, in the sense that $\tau_{i+1} \leq \tau_{i} / 2$. This can be seen as follows, using $\sqrt{1+x} \leq 1+x/2$, for all $x \geq 0$.
    \[\tau_{i+1}+1=c_{i+1}\leq\sqrt{c_i}=\sqrt{1+\tau_i}\leq1+\tau_i/2.\]
    Thus, the running time of these iterations, by the stability of the decider, is
    \begin{equation*}
        \sum\nolimits_{i=m}^M O\bigl(T(n, 1+\tau_i)\bigr)
        =
        O\Bigl( \sum\nolimits_{i=m}^M
        {\left(\tfrac{\eps(\tau_i+1) }{\tau_i} \right)}^c\,
        T(n,1+\eps)\Bigr)
        =%
        O\bigl( T(n,1+\eps) \bigr).
    \end{equation*}
    Thus, the overall running time is $\OO( T(n,1+\eps) + T(n,2) \log \log \tfrac{\beta}{\alpha})$.
\end{proof}

\section{Postscript: Approximating the marching \Frechet distance}
\seclab{marching}

We sketch here a simpler algorithm for the special case of approximating the marching \Frechet distance (either strong or weak). As a reminder, in the marching variant, two consecutive configurations in the alignment are generated by moving the location along a sequence by one position forward (or backward in the weak variant).

The basic idea is to use a sequence of simplifications of decreasing complexity so that the complexities form a geometric series with overall linear size and time. This algorithm is a contrast to our main algorithms, where the \Frechet distance might lie in a significantly larger interval, and simplification by itself is not enough, as the complexity of the curves does not decrease.

We only sketch the algorithm as many of the ideas used here are present in previous work \cite{dhw-afdrc-12}, and in our main algorithms.

\subsection{The ``known'' algorithm}

\paragraph*{Simplification.}

For a parameter $\delta >0$, a \emphw{$\delta$-marking} of sequence $P$ is created by scanning the sequence, and marking the first vertex encountered distance larger than $\delta$ from the last marked vertex (the first and last vertices of $P$ are always pre-marked).  The subsequence formed by this $\delta$-marking is the \emphi{$\delta$-simplification} of $P$, denoted by $\simpY{\delta}{P}$. See \cite{dhw-afdrc-12}.  It is easy to verify that for any two sequences $P$ and $Q$, we have (for the marching \Frechet distance) that
\begin{math}
    \cardin{ \dFY{P}{Q} - \dFY{\simpY{\delta}{P}}{\simpY{\delta}{Q}} } %
    \leq%
    2 \delta.
\end{math}

\paragraph*{Bounded spread.}

The \emphi{gap} of a sequence $P$ is the length of the longest edge of $P$, that is
\begin{math}
    \gapX{P} = \max_{i =1}^{\cardin{P}-1} \dY{p_i}{p_{i+1}}.
\end{math}
As the alignment ``must'' walk across the longest gap (in either curve), and one can bound the total length of the curves by multiplying the gap by the number of edges, we get the following.

\begin{lemma}
    \lemlab{goody}%
    Given two discrete curves $P = p_1, \ldots, p_\alpha$ and $Q= q_1, \ldots, q_\beta$ with a total of $n$ vertices, one can compute in linear time a number
    \begin{equation*}
        L = \max\Bigl(
        \frac{\gapX{P}}{2},
        \frac{\gapX{Q}}{2},
        \dY{p_1}{q_1},
        \dY{p_{\alpha}}{q_{\beta}}
        \Bigr),
    \end{equation*}
    such that for the marching \Frechet distance (either weak or strong), we have $L \leq \dFY{P}{Q} \leq 2 n L$.
\end{lemma}

Intuitively, the above implies that for the \fseq discrete \Frechet distance problem, both curves are densely sampled, and the range of distances we have to search over is already relatively small. This makes the problem significantly easier, at least approximately.

\paragraph*{An approximate decider for the marching \Frechet distance.}

We are given two $c$-packed sequences of points $P$ and $Q$ in $\Re^d$. Additional input parameters $\rad, \eps > 0$ are provided. The task at hand is to decide if $\dFmY{P}{Q} < \rad$ or $\dFmY{P}{Q} \geq \rad$. The basic idea is to simplify the curves to resolution $\eps \rad/4$.  Next, one uses \BFS (or a similar graph search algorithm) to explore all the feasible configurations in the parametric discrete space for the two simplified curves, ignoring configurations with elevations exceeding $(1+\eps)\rad$. This approach is standard by now---see \citeau{dhw-afdrc-12} for a similar algorithm, and we thus get the following (note that this works for either the strong or weak variants).

\begin{lemma}
    \lemlab{decider_s_f}
    For the \fseq \Frechet distance for two $c$-packed curves of total complexity $n$, and parameters $\eps > 0$ and $\rad > 0$, one can in $O\bigl( c n(1+\tfrac{1}{\eps})\bigr)$ time, return one of the following values:
    \begin{compactenumI}
        \item ``$ < \rad$'' $\implies$ $\dFY{P}{Q} \leq (1+\eps)\rad$,

        \item ``$ > \rad$'' $\implies$ $\dFY{P}{Q} > \rad$, or

        \item return an interval $\dFY{P}{Q} \in [\alpha, (1+\eps)\alpha]$, where $\alpha$ is some real number.
    \end{compactenumI}
\end{lemma}

The \fseq \Frechet distance is contained in the interval $[L, n2L]$. By \lemref{goody}, one can perform a binary search over this interval to get an $(1+\eps)$-approximation. Specifically, plugging \lemref{decider_s_f} into \thmref{binary_search} yields the following.

\begin{lemma}
    Given two $c$-packed discrete curves $P$ and $Q$ of total complexity $n$, and a parameter $\eps > 0$, the above algorithm $(1+\eps)$-approximates, in $\OO( \tfrac{c}{\eps} n + cn \log \log n)$ time, the marching \Frechet distance (either strong or weak) between the two curves.
\end{lemma}

\subsection{A linear-time constant-approximation algorithm}

The input is two curves $P$ and $Q$ that are $c$-packed with a total complexity of~$n$.  The algorithm sets $P_0 = P$ and $Q_0 =Q$. The algorithm starts with the lower bound $L$ on the marching \Frechet distance from \lemref{goody}, and let $\rad_1= L$. The algorithm sets $\delta_1 = \epsA \rad_1$, where $\epsA = 1/32$.

In the $i$\th iteration, for $i>0$, the algorithm computes $P_i = \simpY{\delta_i}{P_{i-1}}$ and $Q_i = \simpY{\delta_i}{Q_{i-1}}$. It then uses the decider procedure on $P_i$ and $Q_i$ to decide if (approximately) $\dFY{P_i}{Q_i} \leq 2\rad_i$, by running the decider with parameter $1+\epsA$. If the procedure returns that this is indeed the case, then the marching \Frechet distance between $P$ and $Q$ is (say) at most $4\rad_i$ (see \lemref{correct} below for details), and at least $\rad_{i-1}$. In this case, the algorithm starts the second stage, of converting this into a $(1+\eps)$-approximation, as described below. Otherwise, the algorithm sets $\rad_{i+1} = 2 \rad_i$ and $\delta_{i+1} =2 \delta_i$, and continues to the next iteration.

\paragraph*{Analysis.}

In every iteration, the algorithm continues only if the marching \Frechet distance is significantly bigger than the current candidate distance. Thus, the algorithm is allowed to double the simplification radius. As such, the simplification radii form a geometric series that is bounded by the true marching \Frechet distance. A careful tracking of the constants involved leads to the following.

\begin{lemma}
    \lemlab{correct}%
    The algorithm returns a value $\rad_i$, which is a constant approximation. That is, the marching \Frechet distance is in the interval
    \begin{math} [\tfrac{3}{8} \rad_{i},\tfrac{5}{4} \rad_i].
    \end{math}
\end{lemma}

\begin{lemma}
    The running time of the algorithm is $O\bigl(cn \bigr)$.
\end{lemma}
\begin{proof}
    Clearly, the $i$\th iteration takes $O\bigl( c(|P_i| + |Q_i| ) \bigr)$ time.  The total length of $\lenX{P_0} \leq (n-1) \gapX{P} \leq 2n \rad_0$ and similarly $\lenX{Q_0} \leq 2n \rad_0$. Furthermore, the distance between any two consecutive vertices of $P_0$ (or $Q_0$) is at most $\rad_0$. The curve $P_i$ is the result of a sequence of simplifications of $P_0$. Still, any pair of consecutive vertices is at a distance of at least $\delta_i$ from each other in the original curve. Thus, we have
    \[
        \cardin{P_i} \leq 4 + \frac{\lenX{P}}{\delta_i} \leq 4 + \frac{\lenX{P}}{2^{i-1}\gapX{P} /64 } = O(1 + n/2^{i-1} ).
    \]
    It readily follows that $\sum_i |P_i| = O( n)$ and $\sum_i |Q_i| = O(n)$, establishing the running time.
\end{proof}

\paragraph*{$(1+\eps)$-approximation.}

To boost the constant approximation, we again plug \lemref{decider_s_f} into \thmref{binary_search} to search over this constant spread interval.

\begin{theorem}
    \thmlab{marching}%
    Given two $c$-packed discrete curves $P$ and $Q$ of total complexity $n$, and a parameter $\eps > 0$, the above algorithm $(1+\eps)$-approximates, in $O\bigl( c(1 + \tfrac{1}{\eps}) n \bigr) $ time, the marching \Frechet distance (either weak or strong) between $P$ and $Q$.
\end{theorem}

\section{Conclusions}
\seclab{conclusions}

We presented several improved algorithms for approximating the \Frechet distance between continuous and discrete $c$-packed curves. In particular, our running time is $O(n)$ instead of $O(n \log n)$ if $c,\tfrac{1}{\eps} \in O(1)$.  The previous algorithm of \citeau{dhw-afdrc-12}, and follow-up work, retained the extra $\log$ factor because they performed a logarithmic number of calls to a linear time decider. The decider time was, in turn, only linear because the input curves were $c$-packed. A natural approach to get a linear-time algorithm is a gradation-type sequence of simplifications, as done in the simpler discrete algorithm for the marching distance; see the algorithm in \secref{marching}.

However, this approach does not work for the more prevalent variants of the \Frechet distance, as the complexity of the free space, which is the sublevel set of the elevation function used by the decision algorithm, does not become simpler quickly enough. Instead, the main insight is that one can compute an approximate elevation function, reminiscent of a bipartite well-separated pair decomposition, once and for all in linear time. This approximate elevation function can now be used to perform the search for a path realizing the \Frechet distance. The challenge is that the natural exploration algorithms use a heap, which would mitigate the benefit so far. Furthermore, the new structure is no longer simply a grid, but rather a BSP (binary space partition) of the parametric space.  To handle this, the algorithm uses specific properties of the elevation function to carefully propagate the information across the space as it performs an exploration for the pass (i.e., the bottleneck distance). The process becomes more challenging as one needs to respect the monotonicity of the desired alignment.  An interesting side result is that the strong continuous \Frechet distance between two 1-dimensional curves ($c$-packed or otherwise) can be computed exactly in $O(n^2)$ time (\corref{1dexact}).  The algorithm is quite simple and does not use the data structure of \secref{decomposing-free-space}. Instead, it uses the elevation functions of the cells directly.

\bigskip

In summary, we consider the result in this paper to be quite surprising --- that one can, even approximately, get a linear-time algorithm stands in contrast to the numerous known algorithms for \Frechet distance and its variants, all requiring superlinear running time (and usually at least quadratic).

\paragraph*{It's deja v{}u all over again.}

Our techniques are reminiscent of ideas found in the literature. The parametric space decomposition is similar to the one used by \cite{ypfa-seaad-16, afpy-adtwe-16}, except that they get an extra log factor in their analysis, which our tighter analysis avoids.  A similar decomposition is implicitly used by \citeau{bkn-wdfpa-21}---but this is an applied paper, and they only solve the decision problem, so while in practice their algorithm is (very) fast, the theoretical performance is unclear.

\citeau{bblmm-cfdrl-16}, for the case of polyhedral distance functions,
removes the extra $\log$ factor (the algorithm is still quadratic in this
case) using the nicer behavior of the distance function being
propagated. This is somewhat similar conceptually to the approximation
simplified elevation function used by the algorithm here.

\printbibliography

\end{document}